\documentclass[11pt]{article}

\usepackage[T1]{fontenc}

\usepackage{graphicx}
\usepackage{amsmath, amssymb, amsthm}
\usepackage{thmtools}

\usepackage{parskip}
\usepackage{amsmath, amssymb, amsthm}

\usepackage[backref=page,linktocpage=true,colorlinks=true,linkcolor=blue,citecolor=blue,urlcolor=blue]{hyperref}

\usepackage{cleveref}
\usepackage{comment}
\usepackage{xspace}

\usepackage{algorithm}
\usepackage[noend]{algpseudocode}

\usepackage{enumitem}  
\usepackage[margin=1in]{geometry}  

\usepackage{tikz}
\usetikzlibrary{arrows.meta,decorations.pathreplacing}
\usetikzlibrary{calc}
\usepackage[font=scriptsize,justification=centering]{caption}

\usepackage[most]{tcolorbox}

\newtheorem{theorem}{Theorem}[section]
\newtheorem{lemma}[theorem]{Lemma}

\newtheorem{claim}[theorem]{Claim}

\newtheorem{fact}[theorem]{Fact}
\theoremstyle{definition}
\newtheorem{definition}[theorem]{Definition}
\newtheorem{example}[theorem]{Example}

\title{Online Graph Balancing and the Power of Two Choices}
\author{
Nikhil Bansal\thanks{University of Michigan. \texttt{bansaln@umich.edu}.
Supported in part by NSF awards CCF-2327011 and CCF-2504995.} 
\and
Milind Prabhu\thanks{University of Michigan. \texttt{milindpr@umich.edu}. Supported by NSF award CCF-2327011.} 
\and
Sahil Singla\thanks{School of Computer Science, Georgia Tech. \texttt{ssingla@gatech.edu}. Supported in part by NSF awards CCF-2327010 and CCF-2440113.}
\and
Siddharth M. Sundaram\thanks{School of Computer Science, Georgia Tech. \texttt{ssundaram38@gatech.edu}. Supported in part by NSF awards CCF-2327010 and CCF-2440113.}
}

\usepackage{tikz}
\usetikzlibrary{arrows.meta,positioning,calc,decorations.pathreplacing}

\newcommand{\Poi}{\mathtt{Poi}}

\newcommand{\poly}{\ensuremath{\mathsf{poly}}}
\newcommand{\E}{\mathbb{E}}

\renewcommand{\r}{\rangle}

\newcommand{\ignore}[1]{{}}

\newcommand{\opt}{\mathrm{OPT}}

\newcommand{\ggreedy}{G_{\mathrm{greedy}}}

\newcommand{\maxdensity}{\rho^*}

\newcommand{\avgdegree}{d^\mathrm{av}}

\newcommand{\skewness}{\mathrm{Skew}}
\newcommand{\sk}{S}

\newcommand{\logskew}{log-skewness}
\newcommand{\LogSkew}{Log-Skewness}

\newcommand{\skewbir}{skew-biregular}
\newcommand{\SkewBir}{Skew-Biregular}

\newcommand{\numclasses}{h}

\begin{document}
\date{}

\maketitle

\begin{abstract}
In the classic \emph{online graph balancing} problem, edges arrive sequentially and must be oriented immediately upon arrival, to minimize the maximum in-degree. For adversarial arrivals, the natural greedy algorithm is \(O(\log n)\)-competitive, and this bound is the best possible for any algorithm, even with randomization. We study this problem in the i.i.d.\ model where a base graph \(G\) is known in advance and each arrival is an independent uniformly random edge of \(G\). This model generalizes the standard \emph{power-of-two choices} setting, corresponding to \(G=K_n\), where the greedy algorithm achieves an \(O(\log\!\log n)\) guarantee. We ask whether a similar bound is possible for arbitrary base graphs.

\medskip

While the greedy algorithm is optimal for adversarial arrivals and also for i.i.d.\ arrivals from regular base graphs (such as $G=K_n$), we show that it can perform poorly in general: there exist mildly irregular graphs $G$ for which greedy is $\widetilde{\Omega}(\log n)$-competitive under  i.i.d.\ arrivals. In sharp contrast, our main result is an $O(\log\!\log n)$-competitive online algorithm for every base graph~$G$; this is optimal up to constant factors, since an $\Omega(\log\!\log n)$ lower bound already holds even for the complete graph $G=K_n$. The key new idea is a notion of \emph{\logskew{}} for graphs, which captures the irregular substructures in $G$ that force the offline optimum to  be large. Moreover, we show that any base graph can be decomposed into ``\skewbir{}'' pieces at only $O(\log\!\log n)$ scales of \logskew{}, and use this to design a decomposition-based variant of greedy that is $O(\log\!\log n)$-competitive.
\end{abstract}

\bigskip

\setcounter{tocdepth}{1}


\newpage

\allowdisplaybreaks

\section{Introduction}

In the \emph{online graph balancing} problem, edges of an unknown graph on $n$ vertices arrive sequentially. Upon each arrival, the algorithm must immediately orient the edge toward one of its endpoints so as to minimize the \emph{maximum vertex load}, i.e., the largest in-degree of any vertex. 
Motivated by applications in online load balancing and scheduling, this problem has been studied extensively since the 1990s. 
When the entire graph is known in advance, the minimum achievable load equals the \emph{maximum density}---the ratio of edges to vertices---over all 
subgraphs~\cite{hakimi1965degrees}. 
In contrast, when edges arrive online in adversarial order, the classic result of~\cite{ABK-FOCS92}  shows that the \emph{greedy algorithm}, which always orients an edge toward the lesser-loaded endpoint, achieves an $O(\log n)$-competitive ratio, i.e., its maximum load is within an $O(\log n)$ factor of the optimum. 
This logarithmic bound is the best possible for \emph{any} algorithm, possibly randomized, under \emph{adversarial} arrivals.
Over the past three decades, $O(\log n)$-competitive algorithms have also been obtained for several generalizations of online graph balancing~\cite{aspnes1997line,Caragiannis-SODA08,KMS-SODA23,KMPS-FOCS25}.

In many online applications, however, inputs are not adversarial, but instead arise from an underlying stochastic process. 
A standard way to move beyond worst-case analysis in such settings is the \emph{i.i.d.\ model}, where arrivals are drawn independently and identically from a distribution. Again, the benchmark here is the expected value of the offline optimum, or equivalently, the expected maximum density of the \emph{sampled graph}. In general, study of the i.i.d.\ model has led to significantly better competitive ratios for both maximization (e.g., online matchings, combinatorial auctions) and minimization (e.g., facility location, Steiner tree, metric matching) problems where worst-case guarantees are overly pessimistic (see the book~\cite{Roughgarden-BeyondBook}). 

This motivates the following question:

\begin{quote}
\emph{What is the minimum achievable competitive ratio for online graph balancing when edges are sampled i.i.d.\,(say, uniformly at random) from some adversarially chosen base graph?}  
\end{quote}

Despite being a natural question,  progress on this question has been difficult because it vastly generalizes the celebrated ``Power-of-Two-Choices'' model.

\medskip  \textbf{Power of Two Choices.} 
Arguably the ``simplest'' case of the above question is where the arrivals are i.i.d.\ edges of the complete graph $K_n$. This case is precisely the \emph{Power-of-Two-Choices} model, first studied by Azar, Broder, Karlin, and Upfal~\cite{azar1994balanced}. They showed that when each edge is oriented {\em greedily}, the maximum load after $n$ arrivals is only $O(\log\!\log n)$, and moreover, 
no online algorithm can achieve  $o(\log\!\log n)$ maximum load. A simple computation also shows that the expected offline optimum is $O(1)$ (in fact exactly $1$ with probability $1-o(1)$). Together, this implies that $o(\log\!\log n)$-competitive ratios are information-theoretically impossible even under i.i.d.\ arrivals.

In a major advance, and motivated by various applications, Kenthapadi and Panigrahy~\cite{kenthapadi2006balanced}
studied the power-of-two-choices for general base graphs $G$ beyond $K_n$. They called this ``Graphical Allocation'', which is identical to the problem we consider here, and solved it completely for all {\em regular} graphs.
In particular, they showed that for every $d$-regular graph $G$ on $n$ vertices, with degree parameterized as $d=n^{\epsilon}$, for $n$ i.i.d.\,arrivals, the optimum offline load is about $1/\epsilon$ (ignoring lower order terms) and the greedy algorithm achieves maximum load $O(1/\epsilon + \log\!\log n)$, implying an $O(\log\!\log n)$ competitive ratio. This competitive ratio also extends to an arbitrary number of arrivals.
To show this, \cite{kenthapadi2006balanced} used elegant witness tree arguments together with several clever ideas required to handle various dependencies that arise when considering arbitrary regular graphs.

\medskip  \textbf{General (Irregular) Graphs.} 
It is quite remarkable above that for any regular graph $G$, both the offline optimum and the greedy load only depends on its degree $d$ --- irrespective of its specific structure or properties like its connectedness and expansion. 
For general \emph{irregular} graphs, however, \cite{kenthapadi2006balanced} already observed that the situation is far more complex. In particular, they show that when $G$ is the complete 
bipartite graph $K_{n,\sqrt{n}}$, the standard arguments based on density and degree only give $\Omega(1)$ lower bounds, but the offline optimum is $\Omega(\log n/\log\!\log n)$ for a rather subtle reason. We will elaborate more on this later, and understanding this phenomenon will be one of our key contributions.

Indeed, even though the power-of-two-choices paradigm has been studied extensively since \cite{azar1994balanced}, and led to several remarkable techniques and results (see Section \ref{sec:prev-work}), it remains poorly understood for general irregular graphs. To the best of our knowledge, all existing results in the area crucially rely on \emph{regularity}, and often just work with the complete graph $K_n$.

\subsection{Our Results and High-Level Overview}\label{sec:mainResults}

In this work, we fully resolve the question for arbitrary graphs.
Our main result is the following.
\begin{theorem}\label{thm: polyloglog-comp}
For any base graph $G$ on $n$ vertices, there is an $O(\log\!\log n)$-competitive algorithm for online graph balancing, for any arbitrary number of i.i.d.\,arrivals sampled uniformly from $E(G)$.
\end{theorem}

As the $\Omega(\log\!\log n)$ lower bound  already holds for $G = K_n$, as shown by \cite{azar1994balanced}, \Cref{thm: polyloglog-comp} is optimal up to constant factors. 

Perhaps surprisingly, even though Greedy is 
optimal both for adversarial arrivals and for i.i.d.\ arrivals from a regular graph~\cite{kenthapadi2006balanced},\footnote{In fact, \cite{kenthapadi2006balanced} show that Greedy is $O(\log\!\log n)$ competitive even when all the degrees are within an $O(1)$ factor.} it can perform very poorly even for mildly irregular graphs $G$, and hence does not suffice to prove Theorem \ref{thm: polyloglog-comp}.
In particular, we show the following in Section~\ref{sec: greedy-lb}.
\begin{theorem}\label{thm: greedy-lb}
There is a graph~$G$ for which the greedy algorithm, even with random tie-breaking, has maximum load $\Omega(\log n/\log\!\log n)$ after $n$ arrivals, 
while the expected offline optimum is only $O(\log\!\log n)$.
Moreover,  $G$ is only mildly irregular with 
all vertex degrees in the range $[\sqrt{n},\sqrt{n} \log^{3}n]$.
\end{theorem}


The proof of \Cref{thm: polyloglog-comp} rests on four key ideas. First, we identify the right obstruction that makes the offline optimum large. Second, we isolate a class of imbalanced subgraphs on which Greedy nevertheless succeeds. Third, we show that every graph can be decomposed into only $O(\log\!\log n)$ such pieces. 
Fourth, based on this decomposition, we give a new algorithm called {\em Threshold-Greedy} and show that it is $O(\log\!\log n)$-competitive.
We elaborate on these ingredients next.

\smallskip
\noindent
\textbf{(1) \LogSkew{}.}
The example of~\cite{kenthapadi2006balanced} already shows that the presence of imbalanced subgraphs inside $G$ can greatly increase the offline optimum. To capture this effect, we introduce a new graph parameter, which we call \emph{\logskew{}} (see \Cref{subsec: skewed-lb}). Informally, \logskew{} measures how much a bipartite subgraph can be larger on one side than the other while still carrying substantial degree. 
We show that if $\skewness(G)$ is the maximum \logskew{} over all subgraphs of $G$, then the offline optimum is $\Omega(\skewness(G))$. In fact, 
this bound, together with the standard lower bounds based on degree and density, determines the offline optimum up to an $O(\log\!\log n)$ factor. 


\smallskip
\noindent
\textbf{(2) Skew-Biregular Subgraphs.} 
The next question is algorithmic: on what kinds of imbalanced graphs can Greedy perform well? In general, the answer is not all of them: as shown by the example in \Cref{thm: greedy-lb} (both the offline optimum and skewness for it are only $O(\log\!\log n)$). However, we identify a large class of bipartite graphs, which we call \emph{$f$-\skewbir{}} subgraphs (see \Cref{sec:skewbir}), for which Greedy is $O(\log\!\log n)$-competitive, even though the graphs are highly imbalanced. Roughly speaking, in such a graph every vertex on one side has degree at most $d/f$, while every vertex on the other side has degree at most $d f^{\skewness(G)}$. The key point is that high-degree vertices on one side are only incident to proportionally lower-degree vertices on the other, which keeps the associated branching process in the witness tree argument subcritical.

\smallskip
\noindent
\textbf{(3) Graph Decomposition.}
To handle general graphs, our main structural result is an almost-linear-time decomposition showing that every graph can be edge-partitioned into only $O(\log\!\log n)$ $f$-\skewbir{} subgraphs, each capturing one such scale of $f$ (see \Cref{subsec:graphDecomp}). 
The key structural insight which we use is that bounded \logskew{} forces an expansion property, which we use to iteratively peel off edge-disjoint \skewbir{} subgraphs at different scales.  Running Greedy separately on each piece already yields a non-trivial $O((\log\!\log n)^2)$-competitive algorithm. 


\smallskip
\noindent
\textbf{(4) Threshold-Greedy.}
Finally, to obtain the optimal $O(\log\!\log n)$ ratio, we combine this decomposition with a carefully designed {\em Threshold-Greedy} rule that handles all pieces simultaneously without losing the extra factor of $\log\!\log n$ (see \Cref{sec: algo}). Roughly, these thresholds help limit the interaction between the skew-biregular graphs at different scales, without splitting the graph into separate pieces. The analysis of Threshold-Greedy requires a novel witness-tree
argument that exploits the structural properties of the \skewbir{} graphs
produced by the decomposition. 





\medskip
{\bf Organization.}
Section~\ref{sec: setup}  develops the
 preliminaries and reduces the problem to the
bipartite case.
Section~\ref{sec:logskew}  shows that \logskew{} lower bounds the offline optimum, and then uses it to decompose the base graph into a small number of \skewbir{} subgraphs. 
Section~\ref{sec: algo} presents the $O(\log\!\log n)$-competitive  \emph{Threshold-Greedy} algorithm. Section~\ref{sec: greedy-lb} exhibits a base graph on which {Greedy} performs poorly. 
Finally,  Section~\ref{sec:conclusion} concludes with a discussion of
 future directions.

\subsection{Further Related Work}
\label{sec:prev-work}
There has been extensive work on both online graph balancing and the power of two choices. We only describe these briefly here 
and refer to the surveys \cite{sitaraman2001power,mitzenmacher1996power,wieder2017hashing, azar2005line} for more.

\emph{Random-order graph balancing.} 
There have been several attempts to go ``beyond-worst-case'' for online graph balancing, particularly under \emph{random-order} arrival of edges. 
 Already  in 1995, \cite{BroderFLPR-IPL95} showed that Greedy is $\Omega(\log n/\log\!\log n)$-competitive in the random-order setting. Although $(1+\epsilon)$-competitive algorithms are  possible with an additive logarithmic loss in the random order model \cite{GuptaM-MOR16,molinaro2017online}, the interesting recent work of \cite{im2024online} shows that, without an additive loss, every online algorithm is $\Omega(\sqrt{\log n})$-competitive. Consequently, our $O(\log\!\log n)$-competitive guarantee crucially relies on the difference between i.i.d.\,and random-order models.

\emph{Scheduling with ML Advice.}
Another approach to going beyond worst-case analysis uses \emph{machine-learning advice}, 
where the online algorithm receives predictions of a small number of parameters of the input sequence.  
For online load balancing on $n$ machines, there exist $O(n)$ parameters (e.g., one weight per machine) such that an allocation rule governed by these parameters is $O(1)$-competitive in the fractional setting and $O(\log\!\log n)$-competitive in the integral setting~\cite{LiXian-ICML21,lattanzi2020learning}.  
In stochastic settings like ours, one might hope to infer these parameters directly from the underlying distributions.  
However, these parameters do not concentrate sufficiently, so only yield $\poly\!\log(n)$-competitive ratios for online graph balancing.


\emph{Variants of two-choices.} 
Several variants and extensions of the original model \cite{azar1994balanced} have been studied. For example,~the heavily loaded case with many more balls than bins \cite{berenbrink2000balanced,talwar2014balanced,peres2015graphical}, dynamic settings where balls may be deleted and reinserted \cite{cole1998balls,bansal2022balanced}, parallel allocations \cite{stemann1996parallel,lenzen2019parallel}, and allocations under incomplete information \cite{los2021balanced,los2023balanced}. Several elegant proof techniques such as layered induction \cite{azar1994balanced},  witness trees \cite{czumaj1995shared,cole1998balls,vocking2003asymmetry}, differential equations \cite{mitzenmacher1996power}, stability of Markov processes \cite{berenbrink2000balanced,talwar2014balanced}, and potential functions \cite{peres2015graphical, los2023balanced} have also been developed. Our proofs are based on
the witness tree technique, combined with various extensions that are required to handle irregular graphs.

 Two choice model on general graphs (aka graphical allocation) has also been extensively studied since the work of Kenthapadi and Panigrahy \cite{kenthapadi2006balanced}. Godfrey \cite{godfrey2008balls} and Greenhill et al.\ \cite{greenhill2020balanced} extended \cite{kenthapadi2006balanced} to structured families of hypergraphs. Peres, Talwar, and Wieder \cite{peres2015graphical} investigated 
expanders in the heavily loaded regime, which was recently generalized  to regular graphs by Bansal and Feldheim \cite{bansal2022power}.
As discussed before, all these results only consider  \emph{regular} graphs.

\section{Preliminaries and Preprocessing}
\label{sec: setup}
We begin with the formal problem definition and some basic notation. Next, we describe some simple properties of the offline optimum. Finally, we describe a useful preprocessing, which will allow us to assume that the base graph $G$ is bipartite and with bounded degrees on one side.

\subsection{Problem Definition}
We consider the following stochastic graph balancing problem.
\begin{definition}[Graph Balancing]
Given a {\em base} graph $G = G(V,E)$ 
on $n$ vertices,  
at each time step $t=1,\ldots,T$, an edge $e_t$ is drawn uniformly from $E$, with replacement, and
must be (irrevocably) oriented towards one of its endpoints. The load $L(u)$ of a vertex $u$ is the number of edges directed into $u$, and the goal is to minimize the \emph{maximum load}, $M:=\max_{u \in V} L(u)$, over the vertices.
\end{definition}

This is equivalent to 2-choice balls into bins problem studied by \cite{kenthapadi2006balanced}---each vertex corresponds to a bin, and at each time $t$ the two bin choices for the arriving ball are given by the random edge $e_t$. 
Clearly, the case of $G=K_n$ corresponds to the classical 2-choice model.
The case of regular graphs $G$ was solved completely by \cite{kenthapadi2006balanced}.

We will use standard competitive analysis.
Let $G'$ denote the random {\em sampled} (multi)-graph formed by the $T$ sampled edges $e_1, \ldots, e_T$. 
For a fixed sample $G'$, let $M^*(G')$ denote the optimal offline load for $G'$. Similarly, for an online algorithm $A$, let $M^A(G')$ denote the maximum load under \text{A}.\footnote{Strictly speaking, the online load $A(G')$ can depend on the order of arrival of the edges in $G'$. So $G'$ can be viewed as an online sequence.}
For a base graph $G$ and sequence length $T$, let
\begin{equation}
    \label{eq:opt-defn}
    \opt(G,T) := \E_{G'}[M^*(G')]  \quad \text{ and }  \quad A(G,T) := \E_{G'}[M^A(G')]
\end{equation}   denote the expected optimum maximum load and that under $A$, respectively. 

We say that $\text{A}$ is \emph{$\alpha$-competitive} if 
$A(G,T) \leq \alpha \, \opt(G,T)$
for all graphs $G$ and lengths $T$.

\medskip 
{\bf Focusing on $T=n$ arrivals.}
As we are interested in the (multiplicative) competitive ratio, the hardest case is when $T \approx n$. 
Indeed, in \Cref{sec: reductions} we will formally prove that the case of general $T$ arrivals  can be reduced to the case of $T=n$.
So, from now on until \Cref{sec: reductions}, we will assume that $T=n$ and use $\opt(G)$ to denote $\opt(G,n)$, and focus on proving our main result Theorem \ref{thm: competitive-ratio} in the case. We rewrite this below for future reference.

\begin{theorem}\label{thm:mainArrivalsn}
There is an $O(\log\!\log n)$-competitive online algorithm for Graph Balancing problem for any $G$ with $T=n$ uniformly drawn i.i.d.\,arrivals.
\end{theorem}
We now note some basic properties of $\opt(G)$.
\subsection{Simple Lower Bounds on \texorpdfstring{$\opt$}{opt}}
\label{sec: understanding-opt}
Let $H = H(V_H, E_H)$ be a multigraph.
We use $\rho(H): = |E_H|/|V_H|$ to denote the density of $H$. For  $S \subseteq V_H$, let 
 $H[S]$ denote the induced subgraph of $H$ on $S$. 
The max-density of $H$ is defined as 
\[ \maxdensity(H) := \max_{S \subseteq V_H} \rho(H[S]).\]
Suppose $H$ corresponds to some set of edge arrivals. 
It is a classical result   
 \cite{hakimi1965degrees}  that 
$M^*(H)= \lceil \maxdensity(H) \rceil$.
(The lower bound is easy to see as $M^*(H) \geq \rho(H[S])$ for any subset $S$, as every edge of $H[S]$ must be oriented to some vertex in $S$).  
Note that $\maxdensity(H)\geq 1/2$ for any non-empty graph (just consider a single edge), so we will ignore the ceiling in  $M^*(H)= \lceil \maxdensity(H) \rceil$ henceforth.
 Even though $\maxdensity(H)$ depends on all subsets $S$, it can be computed efficiently using maximum flow.\footnote{In fact, repeatedly picking the smallest degree vertex in $H$, orienting all edges into it, and removing it, gives a simple $2$-approximation to $\maxdensity(H)$, which will suffice for us.}
  
So for our problem, $\opt(G)$ is precisely the expected max-density of the sample $G'$, i.e.,
\begin{equation}
    \label{eq:opt}
\opt(G) = \E_{G'} [M^*(G') ] = \E_{G'} [\maxdensity(G')] = \E_{G'} \left[ \max_{S \subseteq V} \rho(G'[S])\right].
\end{equation}

Let $\avgdegree(G) = 2|E|/n$ denote to the average degree of $G$. To avoid trivialities, we will assume that $\avgdegree(G)\geq 1$.
Recall that $G'$ is obtained by sampling  $n$ edges from $E$ (as we assume $T=n$). Thus the number of occurences in $G'$ of an edge $e\in E$ follows a binomial distribution $\text{Bin}(n,p)$, with $p=1/|E|$. In particular, the expected number of occurences of $e$ is $np =2/\avgdegree(G)$.

We have the following simple lower bounds on $\opt$
in terms of its max-density $\maxdensity(G)$ and $\avgdegree(G)$.
 \begin{claim}
\label{cl:simple-lb-3}
For any base graph $G$, the expected optimum load $\opt(G)$ satisfies
\begin{align}
\label{lb:max-density}
    \opt(G) & \geq 2\maxdensity(G)/\avgdegree(G) \qquad \qquad & \text{(density lower bound)}\\
\label{lb:degree}
    \opt(G) & = \Omega\left( \log n/\log (\avgdegree(G) \cdot \log n) \right) \qquad & \text{(edge-multiplicity lower bound)}
\end{align}
\end{claim}

The above claim follows simply from the linearity of expectation and standard concentration results for the binomial distribution. We include its proof for completeness in \Cref{sec:bin-conc}. 

\subsection{Preprocessing to Left-Degree-Bounded  Bipartite Graphs}
We now describe a useful processing step to simplify the structure of $G$.

\smallskip
{\bf Preprocessing.} Let $G=(V,E)$ be an arbitrary base graph with
max-density
$\maxdensity(G)$.
 Compute an orientation of edges in $E$ with maximum in-degree $\lceil\maxdensity(G)\rceil$.

Construct a (undirected) bipartite graph $H=(V_L,V_R,E_H)$ from this orientation as follows: For each $u\in V$, add a vertex $u_L \in V_L$ and $u_R \in V_R$.
For each edge $e=(u,v)\in E$, add a corresponding edge $e_h$ in $E_H$ where $e_h=(u_L,v_R)$ 
if $e$ is directed towards $u$, else $e_h=(v_L,u_R)$.  

We have the following useful property.

\begin{lemma}\label{clm: left-degree-bounded}
For any base graph $G$ and any number of arrivals $T$, 
an $\alpha$-competitive graph balancing algorithm on the corresponding  base graph $H$ implies a $2\alpha$-competitive algorithm on $G$.
Moreover, $\maxdensity(H)\leq \maxdensity(G)\leq 2 \maxdensity(H)$, and 
the maximum left-degree  $\Delta_L(H) := \max_{v\in V_L} d(v) \leq 4\, \maxdensity(H)$.
\end{lemma}
\begin{proof}
Consider any instance $G'$ on $G$, and let $H'$ be the corresponding instance in $H$. 
Then, given any assignment of $H'$ with max-load $c$, the corresponding assignment in $G'$ (assign edges at $u_L$ and $u_R$ to $u$) has max-load $2c$. Conversely, given any assignment of $G'$ with max-load $c$, assigning edge on $u$ to $u_L$ or $u_R$ in $H'$ (depending on the corresponding orientation in $G$) has max-load $c$. 

For any subset $S\subset V$, the density $\rho(G[S]) = 2\rho(H[S_L,S_R])$, for $S_L,S_R$ in $V_L,V_R$ corresponding to $S$. Conversely, for any $S_L\subset V_L,T\subset V_R$, for corresponding $S,T \subset V$, we have $\rho(H[S_L,T_R]) \leq  \rho(G[S\cup T])$. 
By the property of the orientation, $\Delta_L(H)\leq \lceil \maxdensity(G) \rceil \leq  2\maxdensity(G) \leq 4 \maxdensity(H)$.   
\end{proof}

Also note that $\avgdegree(H) = \avgdegree(G)/2.$  
Henceforth, we will assume that the base graph $G=(L,R,E)$ is bipartite and satisfies $\Delta_L(G) \leq 4 \maxdensity(G)$. We call such graph {\em left-degree-bounded}.

Finally, we can assume the following bound on left-degrees, else the problem becomes trivial.
\begin{claim}
\label{cl:smallish-left-degree}
 We can assume that $G$ has maximum left degree $\Delta_L(G) \leq \log n \cdot  \avgdegree(G)$.
\end{claim}
\begin{proof}
If $\Delta_L(G) \geq \log n \cdot \avgdegree(G)$,
we claim that simply assigning each arriving edge $e=(u,v)$ to its left-endpoint $u$ is $O(1)$-competitive. 
Indeed, for any sample $G'$ of $n$ edges, by Chernoff bounds, with high probability, every left vertex $u$ has degree $O(d_G(u)/\avgdegree(G) + \log n)
= O(\Delta_L(G)/\avgdegree(G))$.
As $G$ is left-bounded-degree
this is at most $O(\maxdensity(G)/\avgdegree(G))$ which is $O(\opt)$ by Claim \ref{cl:simple-lb-3}.
\end{proof}

 \section{\LogSkew{} and Graph Decomposition}\label{sec:logskew}
As discussed in \Cref{sec:mainResults}, highly imbalanced bipartite subgraphs are a
key obstruction to obtaining sharp bounds on $\opt$. Here, we
formalize this obstruction via \logskew{} and show that it lower bounds $\opt$. We then introduce
\skewbir{} graphs in \Cref{sec:skewbir}, a broad class of imbalanced bipartite
graphs for which Greedy is $O(\log\!\log n)$-competitive. In
\Cref{subsec:graphDecomp}, we describe a procedure to decompose any graph $G$ into
$O(\log\!\log n)$ edge-disjoint \skewbir{} subgraphs.

 \subsection{\LogSkew{} and Lower Bound} \label{subsec: skewed-lb}

A key source of lower bounds for the offline optimum is the presence of sufficiently imbalanced bipartite subgraphs. 
We begin with an instructive example that motivates our definition of \logskew{}. 

\begin{example}[Imbalanced biregular subgraph] \label{example:biregular}
Suppose that the base graph $G$ has average degree $\avgdegree$ and contains a bipartite
biregular subgraph $H=(A,B,E_H)$ whose left degree is $\avgdegree/f$ (for some $f\geq 1$) and whose right
degree is $\avgdegree \cdot f^s$ (see \Cref{fig:imbalanced-bireg}).
We claim that this already implies $\opt(G)=\Omega(s)$ (we sketch this below and refer to  Lemma \ref{lem:skew-lb} for a formal argument). 
Indeed, considering only the arrivals from $H$, the sampled degree of each vertex in
$A$ is roughly distributed as a Poisson random variable $\Poi(1/f)$. Hence, up to lower-order factors, a
$1/f^{\Omega(s)}$ fraction of the vertices in $A$ have sampled degree
$\Omega(s)$. As $H$ is biregular, $|A|=f^{s+1}|B|$, so the number of such
vertices is at least $|B|$. These vertices, together with $B$, therefore induce
a sampled subgraph of density $\Omega(s)$, resulting in $\opt(G)=\Omega(s)$.
\end{example}

Observe that the subgraph $H$ above can have much fewer vertices that $G$ and can be {\em hidden} deep inside $G$. Moreover, different values of $f$ can result in the same $\Omega(s)$ lower bound. 
To capture this quantity $s$, we introduce the following definition of \logskew{}.

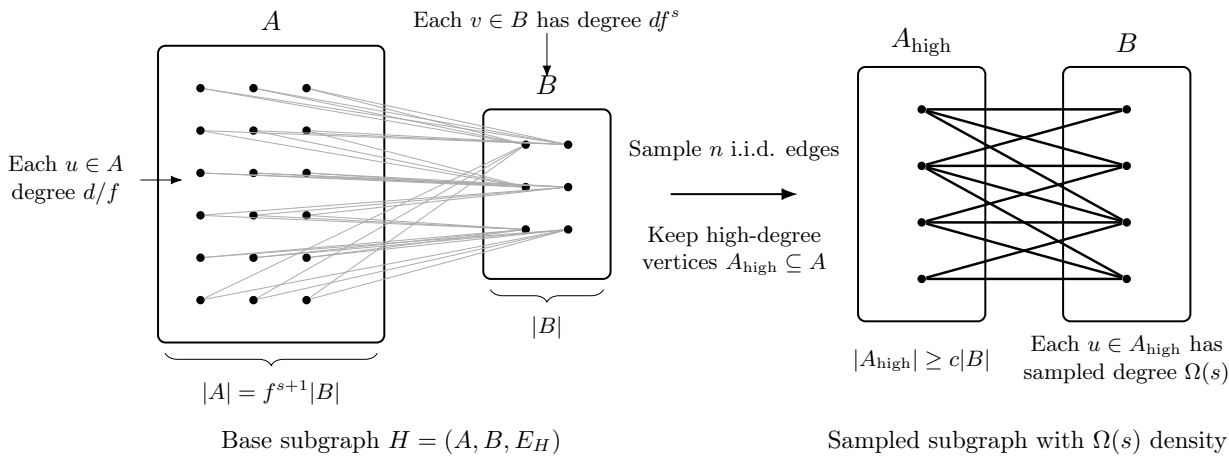
\begin{figure}[t]
\centering
\scalebox{0.94}{
\begin{tikzpicture}[
    x=1cm,y=1cm,
    >=Latex,
    dot/.style={circle, fill=black, inner sep=1.2pt},
    panel/.style={rounded corners=3pt, draw=black, thick},
    lab/.style={font=\small, align=center},
    note/.style={font=\footnotesize, align=center}
]
\node[lab] at (1,-3.5) {Base subgraph $H=(A,B,E_H)$};
\draw[panel] (-2.3,-2.1) rectangle (0.9,2.1);
\draw[panel] (2.3,-1.2) rectangle (4.1,1.2);
\node at (-0.7,2.45) {$A$};
\node at (3.2,1.55) {$B$};

\foreach \y [count=\r] in {1.5,0.9,0.3,-0.3,-0.9,-1.5}{
  \foreach \x [count=\c] in {-1.7,-0.95,-0.2}{
    \pgfmathtruncatemacro{\k}{3*(\r-1)+\c}
    \node[dot] (a\k) at (\x,\y) {};
  }
}

\foreach \y [count=\r] in {0.7,0.1,-0.5}{
  \foreach \x [count=\c] in {2.9,3.5}{
    \pgfmathtruncatemacro{\k}{2*(\r-1)+\c}
    \node[dot] (b\k) at (\x,\y) {};
  }
}

\foreach \u in {1,2,3}{\draw[gray!65, line width=0.35pt] (a\u)--(b1); \draw[gray!65, line width=0.35pt] (a\u)--(b2);}
\foreach \u in {4,5,6}{\draw[gray!65, line width=0.35pt] (a\u)--(b2); \draw[gray!65, line width=0.35pt] (a\u)--(b3);}
\foreach \u in {7,8,9}{\draw[gray!65, line width=0.35pt] (a\u)--(b3); \draw[gray!65, line width=0.35pt] (a\u)--(b4);}
\foreach \u in {10,11,12}{\draw[gray!65, line width=0.35pt] (a\u)--(b4); \draw[gray!65, line width=0.35pt] (a\u)--(b5);}
\foreach \u in {13,14,15}{\draw[gray!65, line width=0.35pt] (a\u)--(b5); \draw[gray!65, line width=0.35pt] (a\u)--(b6);}
\foreach \u in {16,17,18}{\draw[gray!65, line width=0.35pt] (a\u)--(b6); \draw[gray!65, line width=0.35pt] (a\u)--(b1);}

\node[note, anchor=east] at (-2.65,0.2) {Each $u\in A$\\degree $d/f$};
\draw[-Latex, thin] (-2.55,0.2) -- (-1.95,0.2);

\node[note] at (3.2,2.45) {Each $v\in B$ has degree $d f^{s}$};
\draw[-Latex, thin] (3.2,2.28) -- (3.2,1.65);

\draw[decorate,decoration={brace,mirror,amplitude=5pt}] (-2.2,-2.25) -- (0.8,-2.25)
  node[midway,below=7pt,note] {$|A| = f^{s+1}|B|$};

\draw[decorate,decoration={brace,mirror,amplitude=5pt}] (2.4,-1.35) -- (4.0,-1.35)
  node[midway,below=7pt,note] {$|B|$};

\draw[-Latex, thick] (4.95,0.0) -- (6.75,0.0);
\node[note] at (5.85,0.62) {Sample $n$ i.i.d. edges};
\node[note] at (5.85,-0.8) {Keep  high-degree \\ vertices $A_{\mathrm{high}} \subseteq A$ };

\draw[panel] (7.6,-1.8) rectangle (9.4,1.8);
\draw[panel] (10.5,-1.8) rectangle (12.3,1.8);
\node at (8.5,2.15) {$A_{\mathrm{high}}$};
\node at (11.4,2.15) {$B$};

\foreach \y [count=\k] in {1.2,0.4,-0.4,-1.2}{
  \node[dot] (h\k) at (8.5,\y) {};
}

\foreach \y [count=\k] in {1.2,0.4,-0.4,-1.2}{
  \node[dot] (c\k) at (11.4,\y) {};
}

\foreach \u/\v in {1/1,1/2,2/1,2/2,2/3,3/2,3/3,3/4,4/3,4/4,1/3,2/4}{
  \draw[line width=1.0pt] (h\u)--(c\v);
}

\node[note] at (8.5,-2.35) {$|A_{\mathrm{\mathrm{high}}}| \ge c|B|$};
\node[note] at (11.4,-2.35) {Each $u\in A_{\mathrm{high}}$ has\\ sampled degree $\Omega(s)$};
\node[lab] at (10,-3.5) {Sampled subgraph with $\Omega(s)$ density};
\end{tikzpicture}
}
\caption{Imbalanced biregular obstruction motivating log-skewness. A large low-degree side $A$ and a small high-degree side $B$ create, after sampling, a dense core of load $\Omega(s)$.}
\label{fig:imbalanced-bireg}
\end{figure}

\begin{definition}[\LogSkew{}]\label{def:skewness}
Let $G=(V_L,V_R,E)$ be a bipartite base graph on $n$ vertices with average degree $\avgdegree$. 
Let $H = (A, B, E_H)$ be a subgraph of $G$, with  $|A| \ge |B|$, and let
$d_A^{\min} = \min_{u \in A} \deg_H(u)$ be the minimum left-degree  in  $H$. 
We define the \emph{\logskew{}} of $H$ as
\begin{equation}
\label{eq:skew-h}
\sk(H)
:= 
\frac{\log(|A|/|B|)}
{\log(  (\avgdegree \cdot \log^2 n )/d_A^{\min})}.
\end{equation}
The \emph{maximum \logskew{}} of $G$, denoted $\skewness(G)$, is the maximum value of $\sk(H)$ over all subgraphs $H$ of $G$.
\end{definition}

\smallskip

To parse this, let us consider Example \ref{example:biregular} above. Then the numerator in \eqref{eq:skew-h} is $\log (|A|/|B|) = \log f^{s+1} \approx s \log f$, and as the vertices in $A$ have degree $\avgdegree/f$ the denominator is $\approx \log f$ (up to an additive $\log\!\log  n$ term, needed for a technical reason later). So $\sk(H) \approx \log(f^{s+1})/\log f \approx s$.

We now show that $\skewness(G)$ lower bounds $\opt(G)$.

\begin{lemma}\label{lem:skew-lb}
For any base graph $G$, we have $\opt(G)  = \Omega(\skewness(G))$. 
\end{lemma}
\begin{proof}
Fix some subgraph $H=(A,B,E_H)$ for which $\skewness(G)=\sk(H)$ and in which the minimum degree of any left vertex in $H$  is $d^{\min}_A$. Let $H'$ denote $H$ restricted to the sample $G'$. 
We will show that  $\E[\maxdensity(H')]=\Omega(\sk(H))$. This will imply  the result as
$\opt(G) =  \E[\maxdensity(G')] \geq \E[\maxdensity(H')]$.

Let $\lambda= \min (1, d^{\min}_A/\avgdegree(G))$.

Since each edge of $G$ is sampled with probability $2/(n\avgdegree(G))$, the degree $\deg_{H'}(u)$ in $H'$ of any vertex $u\in A$ is distributed as $\text{Bin}(n,p)$ with $p = 2d_H(u)/(n\avgdegree(G))\geq 2\lambda/n$. 
For a random variable $X$ distributed as $\text{Bin}(n,2\lambda/n)$, consider the threshold $t$ such that $\Pr[X \geq t] \geq 2|B|/|A|$. 
Standard concentration for the binomial distribution (Fact \ref{fact:binom} with $\delta= 2|B|/|A|$) gives that 
\[t= 
 \Omega\left(\frac{\log (|A|/|B|)}{\log (1/\lambda) + \log\!\log (|A|/|B|)}\right) = \Omega(\sk(H))\] 
 by the definition of $\sk(H)$ and using  $\log\!\log (|A|/|B|) \leq \log\!\log n$.
 
For $u \in A$, let $Z_u$ be the indicator of the event 
$\deg_{H'}(u) \ge t$, and let $Z = \sum_{u \in A} Z_u$.  
By the choice of $t$, we have $\E[Z] \ge 2|B|$.  
As the $Z_u$ are negatively associated (due to the fixed number of samples $n$ in $G'$), see e.g.,~\cite{dubhashi1996balls}, standard
concentration bounds imply that $Z\geq |B|$ with high probability.
But then the subgraph of $H'$ formed by $B$ and the $|B|$ vertices $u$ in $A$ with the largest $\deg_{H'}(u)$ has density $\Omega(t)$. 
\end{proof}

\subsection{\SkewBir{} Graphs and Greedy}\label{sec:skewbir}

As noted in Theorem \ref{thm: greedy-lb}, another complication with irregular graphs is that Greedy can perform very poorly relative to $\opt$. For example, for the graph $G$ in Theorem \ref{thm: greedy-lb}, $\opt(G)=O(\log\!\log n)$ but Greedy has load $\Omega(\log n/\log\!\log n)$.

However, we identify a broad class of highly imbalanced bipartite graphs, which includes \Cref{example:biregular}, where
Greedy turns out to be $O(\log\!\log n)$-competitive. Roughly,
these are bipartite graphs where vertices on one side may have high-degree, but the vertices on the other have correspondingly low-degree. We formalize this below.

For a bipartite graph $H=(L,R,E_H)$, let $\Delta_L(H)$ and $\Delta_R(H)$
denote the maximum left and right degree respectively. 

\begin{definition}[\SkewBir{} Graphs]
Let $G = (L, R,E)$ be an $n$-vertex left-degree-bounded bipartite graph with
maximum density $\maxdensity$ and max-\logskew{} $\skewness(G)$. For a
parameter $f \geq 1$, a subgraph $H$ of $G$ is called \emph{$f$-\skewbir{}}
if,
\begin{align}\label{eq:skewbir-deg}
\Delta_L(H) = O\!\left(\maxdensity/f\right)
\qquad \text{and} \qquad
\Delta_R(H) = O\!\left(\maxdensity \cdot (f \log n)^{2\skewness(G)}\right).
\end{align}
\end{definition}
This is similar to the setting in Example \ref{example:biregular}, except for two minor differences. First, we only upper bound $\Delta_L(H),\Delta_R(H)$. 
Second, we use the maximum density $\maxdensity$ in \eqref{eq:skewbir-deg} instead of $\avgdegree$. This is just for technical convenience (note that $\maxdensity =\Omega(\avgdegree)$ trivially, and $O(\avgdegree \log n)$ by Claim \ref{clm: left-degree-bounded} 
 anyway).

The following lemma shows that  Greedy is $O(\log\!\log n)$-competitive on such graphs.

\begin{lemma}\label{lem:greedy-skewbir}
If \emph{Greedy} is run on the arrivals from any $f$-\skewbir{} subgraph $H$ of
a base graph $G$, then the load contributed by these arrivals to any vertex is $O( \log\!\log n \cdot \opt(G)
)$.
\end{lemma}
We only sketch the argument  in the special case where $H$ is
biregular, with left degrees $\avgdegree/f$ and right degrees $\avgdegree \cdot f^s$, since the same ideas are developed more fully in the general analysis in \Cref{sec: algo}, specifically in the proof of \Cref{lem: cc}.

We will need the following classical fact about \emph{Greedy}.

\begin{fact}[\cite{ABK-FOCS92}]
\label{fact:greedy}
On any $N$-vertex graph, Greedy is $O(\log N)$-competitive for any arbitrary
(adversarial) sequence of edge arrivals.
\end{fact}

\begin{proof}[Proof sketch for \Cref{lem:greedy-skewbir}]
The key observation is that if we delete the first $O(s)$ edges incident to each left vertex, the sampled
subgraph formed by arrivals from $H$ shatters into connected components of size
$N=O(\poly(\log n))$ with high probability. 
The reason is that each right vertex has about $f^s$ sampled neighbors in
expectation, while each left vertex has sampled degree distributed as
$\Poi(1/f)$ and thus it survives the above pruning with probability at most
$f^{-\Omega(s)}$. So the expected number of surviving neighbors of a right
vertex is at most $O(f^s)\cdot f^{-\Omega(s)} \ll 1$. A standard
branching-process argument now shows the claimed bound on the component sizes.

By \Cref{fact:greedy}, the remaining
edges therefore contribute only $O(\log\!\log n \cdot \opt(G))$ load, and the
deleted edges contribute only $O(s)$ additional load, which by the \logskew{} lower bound of 
\Cref{lem:skew-lb} is $O(\log \!\log n \cdot \opt(G))$.
\end{proof}

\subsection{Decomposition into \SkewBir{} Graphs via \LogSkew{}} 
\label{subsec:graphDecomp}
We now show our key decomposition lemma that any left-degree-bounded bipartite graph~$G$ can be decomposed into only $O(\log\!\log n)$ \skewbir{} subgraphs. 
 The resulting decomposition will be the key structural input to our algorithm in \Cref{sec: algo}.

\begin{lemma}[Decomposition into \SkewBir{} Graphs]
\label{lem: graph-decomp}
Let $G$ be an $n$-vertex left-degree-bounded bipartite graph. Then the edges of
$G$ can be partitioned in almost-linear time into $\numclasses = O(\log\!\log n)$ subgraphs
$G_1,\ldots,G_\numclasses$, where each $G_i$ is $2^{2^i}$-\skewbir{}.
\end{lemma}

We first sketch the idea before giving the formal proof. A key observation is that bounded \logskew{} implies a vertex expansion property: any set of sufficiently high-degree left vertices must have a correspondingly large neighborhood. We exploit this expansion to iteratively extract edge-disjoint \skewbir{} subgraphs at various scales of the parameter $f$. More concretely, for each scale $f$ we consider left vertices whose degree is still too large to belong to an $f$-\skewbir{} piece, and peel off incident edges while ensuring that no right vertex receives too many of them. Finally, it suffices to consider only doubly exponential scales of $f$, leading to $O(\log\!\log n)$ pieces.

\begin{proof}
For each $i \in [\numclasses]$, let $f_i := 2^{2^i}$.
We will decompose $G$ by extracting, in round $i$, a subgraph $G_i$ that will
be shown to be $f_i$-\skewbir{}.

Let $H_0 := G$, and for each $i \ge 1$ let $H_{i} = H_{i-1}\setminus G_i$ denote the residual graph after round $i$.
In each round $i$, we will ensure that the subgraph $G_i$ extracted from $H_{i-1}$ satisfies:
\begin{enumerate}[label=(\roman*)]
    \item $\Delta_L(G_i) \leq 16\maxdensity/f_{i}$ and  $\Delta_R(G_i) \leq 16 \maxdensity(f_{i}\log n)^{2\skewness(G)}$.
    \item The residual graph $H_{i} := H_{i-1} \setminus G_i$ satisfies $\Delta_L(H_{i}) \leq   16\maxdensity/f_{i}^2 = 16\maxdensity/f_{i+1}$. 
\end{enumerate}
Property  (i) simply ensures that
each $G_i$ satisfies the degree 
bounds in \eqref{eq:skewbir-deg} and is $f_i$-\skewbir{}. Property (ii) will be useful in the proof below.
Also, as $f_i=2^{2^i}$, and the left degree of $H_{i}$ is at most $16\maxdensity/f_{i+1} \leq 16n/f_{i+1}$, there are at most $\numclasses = O(\log\!\log n)$ rounds.

We now show how to extract $G_i$  in round $i$, while satisfying (i) and (ii). 
Assume inductively that $\Delta_L(H_{i-1}) \leq 16\maxdensity/f_{i}$ at the beginning of round $i$, and note that the base  
case holds for $i=1$ as  $\Delta_L(H_0) = \Delta_L(G) \leq 16 \maxdensity/f_1 = 4\maxdensity$.

Let $b_i := (f_i\log n)^{2\skewness(G)}$. To extract $G_i$, repeat the following:
\begin{enumerate}
    \item Let $U$ be the set of left vertices with degree more than $16\maxdensity/f_i^2$. If $U = \emptyset$ terminate.
\item
     Else, find a $b_i$-matching between $U$ and the right vertices (i.e., every vertex in $U$ has degree $1$ and all right vertices have degree at most $b_i$) and delete it.
\end{enumerate}
Let $G_i$ be the union of all deleted $b_i$-matchings. 

Assuming that Step 2 above always finds a $b_i$-matching whenever $U\neq\emptyset$,
we claim that conditions (i) and (ii) also hold after round $i$.
Indeed, (ii) holds trivially as each left-degree is at most $16\maxdensity/f_i^2$ when $U=\emptyset$ and thus $H_{i} = H_{i-1}\setminus G_i$ satisfies $\Delta_L(H_{i}) \leq 16\maxdensity/f_i^2 = 16\maxdensity/f_{i+1}$. 
To see (i), observe that by the induction hypothesis $\Delta_L(H_{i-1}) \leq 16\maxdensity/f_i$, and the maximum left-degree decreases by $1$ each time a $b_i$-matching is deleted. Thus at most 
$16\maxdensity/f_i$ matchings can be deleted, and we have 
\[\Delta_L(G_i) \leq 16\maxdensity/f_i \text{ and } \Delta_R(G_i) \leq (16\maxdensity/f_i) \cdot b_i \leq 16\maxdensity \cdot (f_i \log n)^{2 \skewness(G)}.\]

It remains to show that there always exists a $b_i$-matching incident to $U$ whenever $U\neq\emptyset$. To this end, we will use 
the \logskew{} property of $G$ to show
that Hall's condition holds --- that for every subset $U' \subseteq U$, its neighborhood $N(U')$ in $H_{i-1}$ has cardinality at least $|U'|/b_i$. Indeed, consider the subgraph $F$ induced by $(U',N(U'))$. By definition, as every vertex in $U'$ has degree at least $16\maxdensity/ f_i^2 \geq \avgdegree/f_i^2$ (as $\maxdensity \geq \avgdegree/2$). As $\skewness(G)$ is at least the \logskew{} $\sk(F)$ of $F$, plugging in the definition of $\sk(F)$ in \eqref{eq:skew-h} gives
\[ \skewness(G) \geq \frac{\log(|U'|/|N(U')|)}{\log(f_i^2 \log^2 n)}, \]
which upon rearranging gives that  
$
|N(U')| 
      \ge |U'|/(f_i\log n)^{2\skewness(G)} = |U'|/b_i$ as desired.

Finally, note that we can accomplish this decomposition in almost-linear time. First, observe that we can assume $\skewness(G)$ is given, as we can guess it up to a small constant factor. Second, we can compute each $G_i$  using a single max $s-t$ flow computation (rather than iteratively deleting $b_i$-matchings): connect every right vertex of  $H_{i-1}$   to $t$ with edge capacity $16 \maxdensity \cdot (f_i \log n)^{2 \skewness(G)}$, every left vertex $u$ in  $H_{i-1}$  to $s$ with edge capacity equal to $\max\{\deg_{H_{i-1}}(u) - 16\maxdensity/f_i^2, 0\}$, and all edges of $H_{i-1}$ have unit capacities.
\end{proof}

\section{Threshold Greedy Algorithm and Analysis}
\label{sec: algo}
The results of the previous section already yield an
$O((\log\!\log n)^2)$-competitive algorithm: one can run
$O(\log\!\log n)$ independent copies of \emph{Greedy}, one on each
\skewbir{} piece (\Cref{lem:greedy-skewbir,lem: graph-decomp}). In \Cref{subsec: t = n}, we describe an algorithm that improves upon this and
obtains the optimal $O(\log\!\log n)$ competitive ratio for $T=n$ arrivals. In Section \ref{sec: reductions}, we extend this algorithm and the analysis to general $T$, using relatively straightforward ideas.

Throughout this section we use the following notation.
We are given a base graph $G$ on $n$ vertices with average degree $\avgdegree$, maximum density $\maxdensity$, and max-\logskew{} $\skewness(G)$. 
By \Cref{clm: left-degree-bounded}, we may assume that $G = (L, R, E)$ is a  bipartite graph with maximum left degree $\Delta_L\leq  4\maxdensity$. We refer to the vertices in $L$ as left and in $R$ as right vertices. 

\subsection{Algorithm for \texorpdfstring{$T=n$}{T=n}}\label{subsec: t = n}

We begin by applying \Cref{lem: graph-decomp} to decompose the base graph $G$
into $\numclasses = O(\log\!\log n)$ edge-disjoint \skewbir{} graphs
$G_1,\ldots,G_\numclasses$. For each $i\in[\numclasses]$, the graph $G_i$ is a
$2^{2^i}$-\skewbir{} graph. We call an edge $e$ in $G$ a \emph{class-$i$} edge
if $e\in G_i$.

Threshold-Greedy is formally defined in \Cref{alg:threshGreedy}. For each left
vertex $u$, it assigns to $u$ the first $\alpha_i$ class-$i$ edges incident to
it; these are the \emph{threshold edges}. All remaining edges are then assigned
using a greedy rule based only on the load induced by non-threshold edges.

The key idea behind this algorithm is the following. For each class $i$, the \skewbir{} structure of $G_i$ implies that, after deleting the threshold number of
class-$i$ edges incident to each left vertex, the corresponding witness tree
formed by the remaining class-$i$ edges dies out quickly. Crucially, this same
witness-tree decay continues to hold even when the remaining arrivals from all
classes are \emph{combined}, which in turn allows us to run a \emph{single}
greedy algorithm on all the remaining arrivals.

\begin{algorithm}
\caption{Threshold Greedy Algorithm} \label{alg:threshGreedy}
\smallskip
    For each class $i \in [\numclasses]$, define a threshold $\alpha_i$\footnotemark\begin{align}\label{eq: threshold}
        \alpha_i := \Theta\!\left( 
        \left(\frac{\maxdensity}{\avgdegree} + \skewness(G)\right)
        \left( \frac{\log\!\log n}{2^i}+1\right)
        \right).
    \end{align}

    At each time step,
    upon arrival of an edge $e = (u,v)$ with $u\in L$ and $v\in R$:
    
    \medskip
    
    \begin{enumerate}
        \item \emph{Threshold rule.} For $i \in [\numclasses]$, let $\ell_i(u)$ denote the load on $u$ due to class-$i$ edges.  If $e$ is a class-$i$ edge and $\ell_{i}(u) < \alpha_i$, assign $e$ to $u$. 
     Such an edge is a \emph{threshold edge}.
    
     \medskip
        
        \item \emph{Greedy rule.} 
 Otherwise, for a vertex $w \in L \cup R$ let $\ell_g(w)$ denote the total load on $w$ due to non-threshold edges of all classes.      
        If $\ell_{g}(u) < \ell_{g}(v)$, then assign $e$ to $u$, else assign $e$ to $v$. Such an edge is called a \emph{greedy} edge. 
   \smallskip
    \end{enumerate}
\end{algorithm} 
\footnotetext{Strictly speaking, the term $(\log\!\log n )/ 2^i$ is meaningful only for indices $i \le \log\!\log\!\log n$, but we use this expression for all $i$ for notational convenience. For larger $i$, the additive constant in \eqref{eq: threshold} dominates.}

\subsection{Analysis}\label{subsec: analysis}

We now prove the following result.
\begin{theorem}
\label{thm: competitive-ratio}
Threshold Greedy is $O(\log\!\log n)$-competitive 
    for $T=n$ arrivals.
\end{theorem}

To prove 
Theorem \ref{thm: competitive-ratio} we need to show that the load on any vertex is $O( \log\!\log n \cdot \opt )$.
We will bound the load due to threshold edges and greedy edges separately.

The total threshold load at any vertex is bounded trivially.
By design, at any vertex, the load due to threshold edges of class-$i$ is at most $\alpha_i$. 
By the definition of $\alpha_i$ in \eqref{eq: threshold}, and the lower bounds $\opt = \Omega( \maxdensity/\avgdegree)$ and $\opt = \Omega(\skewness(G))$ in \Cref{cl:simple-lb-3} and \Cref{lem:skew-lb} respectively,
we have that $\alpha_i = O(\opt \cdot ((\log\!\log n)/2^i + 1))$. Summing over all classes $i\in [\numclasses]$, the total load due to threshold edges at any vertex is at most 
\[\sum_{i = 1}^\numclasses{\alpha_i}= 
O(\log\!\log n \cdot \opt )\] 
So, our main goal will be to bound  the load due to greedy edges.  
To this end, we will consider the subgraph $\ggreedy$, formed by the greedy edges, and show the following key property.
\begin{lemma}\label{lem: cc} 
  With high probability, every connected component in $\ggreedy$ has size $O(\log^2 n)$. 
\end{lemma}

Together with Fact \ref{fact:greedy}, Lemma \ref{lem: cc} directly implies that the expected maximum load due to the greedy edges is  $O(\log\!\log n \cdot \opt)$, completing the proof of \Cref{thm: competitive-ratio}.

Proving Lemma \ref{lem: cc} will be the crux of the analysis, and will be done in Section \ref{subsec: component}. Note that for general base graphs $G$, the sampled graph $G'$ can have arbitrarily large components. 
So this is where we will crucially exploit the properties of the decomposition and the algorithm.

\subsubsection{Bounding Components of \texorpdfstring{$\ggreedy$}{G-greedy}}\label{subsec: component}
We now focus on proving Lemma \ref{lem: cc}.

As any connected component of size $m$ contains a spanning tree on $m$
vertices, it suffices to show that every subtree of $\ggreedy$ has size
$O(\log^2 n)$. We do so via a witness-tree argument, similar in spirit to
those used in analyses of various power-of-two-choices variants
(cf.~\cite{cole1998balls,vocking2003asymmetry,kenthapadi2006balanced}).
Our setting, however, is considerably more delicate. In prior work on
regular graphs, the sampled graph $G'$ itself has no large components and
is therefore much easier to analyze. Here, by contrast, bounding the
components of $\ggreedy$ requires a careful use of both the structural
properties of the decomposition of the base graph and the definition of
the greedy edges.

The proof proceeds as follows.
We first define a special type of tree called
a ``left-heavy'' subtree, and show that it suffices to show that such subtrees have size $O(\log n)$ in $\ggreedy$. Second, we show that with high probability, any left-heavy tree that appears in $\ggreedy$ is of size $O(\log n)$.

\textbf{Reduction to Left-Heavy Trees.}  A subtree $T$ of $G$ is said to be \emph{left-heavy} if at least half its vertices are left vertices in $G$.
We use $|T|$ to denote its size, i.e., the number of vertices in $T$.
We claim that any large tree $T$ in $\ggreedy$ contains a sufficiently large left-heavy tree $T'$.

\begin{claim}\label{clm: left-heavy}
    With high probability, the following property holds:
    Every subtree $T$ in $\ggreedy$, with size $|T|>1$, contains a left-heavy tree $T'$ of size $|T'|= \Omega(|T|/\log n)$.
\end{claim}
\begin{proof}
By \Cref{cl:smallish-left-degree}, and the fact that every edge in $G$ is sampled $2/\avgdegree$ times in expectation, for all left vertices $u$, $\E[\deg_{G'}(u)] \leq 2 \log n$. Since the degrees have a binomial distribution, it  follows by  Chernoff bounds and a union bound that  all left vertices have degree $O(\log n)$ whp.

We condition on this event.
For any subtree $T$ in $\ggreedy$, consider the subtree $T'$ obtained by deleting all the right leaves from $T$.
For each left vertex in $T$, as at most $O(\log n)$ of its right neighbors are deleted, $T'$ must contain at least $\Omega(|T|/\log n)$ vertices.
Finally, as $T'$ has no right leaves,  $T'$ has at least $|T'|/2$ left vertices and is thus left-heavy.
\qedhere
\end{proof}

\textbf{Bounding the Size of Left-Heavy Trees.}
 By Claim \ref{clm: left-heavy}, to prove Lemma \ref{lem: cc}, it suffices to show that every left-heavy tree in $\ggreedy$ has size $O(\log n)$.  

We will do this by upper bounding the number of left-heavy trees in $G$,
and the probability of their appearance in $\ggreedy$. 
To this end, we will need to count the number of left-heavy trees in a more fine-grained way, that uses the information about the decomposition. 
To capture this, we need the following crucial notion of a \emph{pattern}.

\medskip
{\bf Pattern.} A pattern of size $m$ is a rooted tree with $m$ vertices, where  each vertex is unlabeled, and each of the $m-1$ edges is labeled by a number in $[\numclasses]$.

Consider the decomposition of $G = G_1 \sqcup \ldots \sqcup G_\numclasses$ into edge-disjoint subgraphs, given by \Cref{lem: graph-decomp}. We label each edge $e$ of $G$ by $i\in [\numclasses]$ if $e \in G_i$. So we view $G$ as a labeled graph where each vertex has a unique vertex-label in $[n]$, and each edge has some (non-unique) label in $[\numclasses]$.

We say a rooted subtree $T$ of $G$  has pattern $P$ (see \Cref{fig:pattern-sampled-graph}) if there is a bijective mapping $\sigma: V(P) \rightarrow V(T)$ from vertices of $P$ to those of $T$, such that 
\begin{enumerate}
\item  The root of $P$ is mapped to the root of $T$.  \item  For every edge $(a,b)$ in $P$, $(\sigma(a),\sigma(b))$ is an edge in $T$, and 
the edge-label of $(a,b)$ in $P$ is same as the edge-label of $(\sigma(a),\sigma(b))$ in $T$.
\end{enumerate}
 We will only consider subtrees $T$ in $G$ that are rooted at a right vertex.

Our key lemma is the following.
\begin{lemma}\label{clm: pattern}
For any fixed pattern $P$ with $m$ vertices, the probability that a left-heavy tree  with pattern $P$ appears in $\ggreedy$ is at most $n \cdot (8/ \log n)^{m-1}$. 
 \end{lemma}
We prove Lemma \ref{clm: pattern} in Section \ref{subsec: pattern} using a careful counting argument. But let us first see why this immediately implies Lemma \ref{lem: cc}, and hence Theorem \ref{thm: competitive-ratio}.

\begin{proof}[Proof of \Cref{lem: cc}]
It is well-known that there are at most $4^m$ non-isomorphic, unlabeled, rooted trees on $m$ vertices \cite{knuth1997art}.
As there are at most $\numclasses^{m-1}$ choices for the edge-labels for any such tree, there are at most $4^m \numclasses^{m-1}$ distinct patterns on $m$ vertices.
Thus by Lemma \ref{clm: pattern} and a 
union bound over the possible patterns $P$, the probability that any left-heavy tree of size $m$ appears in $\ggreedy$ is at most $4n \cdot (32\numclasses/\log n)^{m-1} = 4n \cdot  2^{-\Omega(m)}$, as $\numclasses = O(\log\!\log n)$. This implies that, with high probability, any left-heavy tree in $\ggreedy$ is of size $O(\log n)$. 
\end{proof}

\subsubsection{Bounding the Probability of a Pattern in \texorpdfstring{$\ggreedy$}{G-greedy}}\label{subsec: pattern}

Fix a pattern $P$. To prove Lemma \ref{clm: pattern}, we  first bound the number of subtrees of $G$ with this pattern, and then bound the probability of appearance of any such subtree in $\ggreedy$.
 
As we only consider subtrees $T$ in $G$ that are rooted at right vertices,  the root of $P$ always maps to a right vertex in $G$. Hence, assuming the root is at depth $0$,  vertices at even depth (resp. odd depth) in $P$ map to right (resp. left) vertices; we refer to them as the right and left vertices of $P$.

A non-root vertex in $P$ whose edge to its parent has class $i$ is called a \emph{class-$i$ vertex}. Our bounds will depend on the number of class-$i$ left vertices in $P$.

\textbf{The number of subtrees of $G$ with pattern $P$.}
To count these subtrees, we first introduce some notation. For each class
$i \in [\numclasses]$, let $f_i := 2^{2^i}$, and let $\Delta_L(i)$ and
$\Delta_R(i)$ denote the maximum left and right degrees, respectively, of
$G_i$. Then $G_i$ is an $f_i$-\skewbir{} graph and satisfies the degree bounds
in \Cref{eq:skewbir-deg}. Let $x_i$ denote the number of class-$i$ left
vertices in $P$.

We now bound the number of ways in which the vertices of $P$ can be mapped to the vertices in $G$.

Clearly, the root of $P$ can be mapped to a right vertex in $G$ in at most $n$ ways. We now map the other vertices of $P$ 
inductively in a breadth-first order. 
For each edge $(p,v)$ in the pattern (viewed as rooted tree), with parent $p$ and child $v$, we bound the number of choices for mapping $v$ in $G$, given the mapping of $p$.

\,\,\,\,\, (i) If $p$ is mapped to a right vertex, and the edge $(p,v)$ has label $i$ in the pattern $P$, then $v$ must be mapped to some (left) neighbor of $p$ in $G_i$. As $G_i$ has maximum right degree $\Delta_R(i)$, there can be at most $\Delta_R(i)$ choices for mapping $v$.

\,\,\,\,\, (ii) If $p$ is mapped to a left vertex, as  $G$ has maximum left degree at most $4 \maxdensity$,  there can be at most $4 \maxdensity$ choices for mapping $v$.

For each class $i \in [\numclasses]$, case (i) arises exactly $x_i$ times, by definition of $x_i$, and case (ii) arises for the remaining 
$m-1-\sum_{i = 1}^\numclasses x_i$ edges.
Thus the number of subtrees in $G$ with pattern $P$ is at most 
\begin{align}\label{eq:tree-count}
n \cdot (4\maxdensity)^{m - 1-\sum_{i} x_i} \; \cdot \; \prod_{i = 1}^\numclasses (\Delta_R(i))^{x_i} = n \cdot (4\maxdensity)^{m-1} \cdot \prod_{i = 1}^\numclasses (\Delta_R(i)/\maxdensity)^{x_i}.
\end{align}

\textbf{Probability of a pattern $P$ in $\ggreedy$. }
Fix a left-heavy tree $T$ with pattern $P$ in $G$.  We show that the probability of appearance of $T$ in $\ggreedy$ is at most
\begin{align} \label{eq:prob-left-heavy} 
(\maxdensity \cdot \log n/2)^{-({m-1})} \cdot  \prod_{i = 1}^\numclasses   (\Delta_R(i)/\maxdensity)^{-x_i}.
\end{align}
The bounds \eqref{eq:tree-count} and \eqref{eq:prob-left-heavy} immediately imply \Cref{clm: pattern}.  We now focus on proving \eqref{eq:prob-left-heavy}.

To bound this probability, a key observation is the following:
Suppose that $T$ appears in $\ggreedy$. Now consider a class-$i$ left vertex $u$ in $T$. Since $u$ has an edge in $\ggreedy$ of class-$i$, it must be that at least $\alpha_i$ threshold edges of class-$i$ incident to $u$ appeared in the sample $G'$ (see \Cref{fig:pattern-sampled-graph}). 
Our parameters $\alpha_i$ in \eqref{eq: threshold}, are chosen precisely to make the probability sufficiently low. 

\begin{figure}[t]
\centering
\begin{tikzpicture}[
    lvtx/.style={circle, draw=black, fill=black, inner sep=0pt, minimum size=3.0pt, line width=0.6pt},
    rvtx/.style={circle, draw=black, fill=gray!30, inner sep=0pt, minimum size=3.0pt, line width=0.6pt},
    lkey/.style={circle, draw=black, fill=black, inner sep=0pt, minimum size=3.6pt, line width=0.7pt, double, double distance=0.5pt},
    rkey/.style={circle, draw=black, fill=gray!30, inner sep=0pt, minimum size=3.6pt, line width=0.7pt, double, double distance=0.5pt},
    root/.style={circle, draw=black, fill=yellow!35, inner sep=0pt, minimum size=4.4pt, line width=1.0pt},
    redg/.style={draw=red!75!black, line width=2.8pt},
    blueg/.style={draw=blue!75!black, line width=2.8pt},
    greeng/.style={draw=green!55!black, line width=2.8pt},
    redt/.style={draw=red!75!black, line width=1.0pt, dotted},
    bluet/.style={draw=blue!75!black, line width=1.0pt, dotted},
    greent/.style={draw=green!55!black, line width=1.0pt, dotted},
    lab/.style={align=center, font=\small}
]

\begin{scope}[xshift=0cm]
    \node[lab] at (0,-1.7) {\textbf{Pattern}};

    \node[root] (pr) at (0,2.2) {};

    \node[lkey] (pu) at (-1.2,1.2) {};
    \node[lkey] (pw) at ( 1.2,1.2) {};

    \node[rvtx] (pa) at (-1.8,0.2) {};
    \node[rvtx] (pb) at (-0.6,0.2) {};
    \node[rvtx] (pc) at ( 1.8,0.2) {};

    \draw[redg]   (pr) -- (pu);
    \draw[blueg]  (pr) -- (pw);
    \draw[blueg]  (pu) -- (pa);
    \draw[greeng] (pu) -- (pb);
    \draw[redg]   (pw) -- (pc);
\end{scope}

\begin{scope}[xshift=8.5cm]
    \node[lab] at (0,-1.7) {\textbf{Sampled graph}};

    \node[root] (r) at (0,2.2) {};

    \node[lkey] (u) at (-1.2,1.2) {};
    \node[lkey] (w) at ( 1.2,1.2) {};

    \node[rvtx] (x1) at (-1.8,0.2) {};
    \node[rvtx] (x2) at (-0.6,0.2) {};
    \node[rvtx] (x3) at ( 1.8,0.2) {};

    \draw[redg]   (r) -- (u);
    \draw[blueg]  (r) -- (w);
    \draw[blueg]  (u) -- (x1);
    \draw[greeng] (u) -- (x2);
    \draw[redg]   (w) -- (x3);

    \node[lvtx] (a1) at (-2.8,2.2) {};
    \node[rvtx] (a2) at (-3.8,1.4) {};
    \draw[redg]   (a1) -- (a2);

    \node[lvtx] (b1) at (2.8,2.1) {};
    \node[rvtx] (b2) at (3.8,1.3) {};
    \node[lvtx] (b3) at (2.9,0.4) {};
    \draw[blueg] (b1) -- (b2);
    \draw[redg]  (b2) -- (b3);

    \node[rvtx] (c1) at (-1.2,-0.5) {};
    \node[rvtx] (c2) at ( 1.4,-0.9) {};
    \node[rvtx] (c3) at (-2.4,0.6) {};
    \node[rvtx] (c4) at ( 2.6,-0.5) {};

    \draw[redt]  (u) -- (a1);
    \draw[redt]  (u) -- (c1);
    \draw[redt]  (u) -- (c3);

    \draw[bluet] (w) -- (b1);
    \draw[bluet] (w) -- (c2);
    \draw[bluet] (w) -- (c4);
\end{scope}

\end{tikzpicture}
\captionsetup{justification=raggedright,singlelinecheck=false}
\caption{Left: a rooted pattern, with edge colors indicating decomposition
classes. Right: a sampled graph containing a tree of this pattern. Thick
colored edges are greedy, and thin dotted edges are threshold. Black circles
are left vertices and gray circles are right vertices. Each left vertex
incident to a class-$i$ greedy edge has at least $\alpha_i$ class-$i$
threshold edges.}
\label{fig:pattern-sampled-graph}
\end{figure}
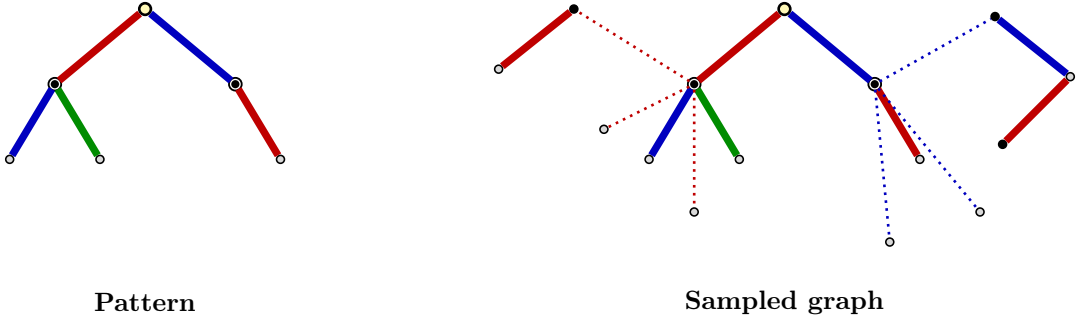

We now give the details.

\textbf{Event $B_T$.} Consider the following event $B_T$, which is necessary for the tree $T$ to appear in $\ggreedy$:

\,\,\,\,\, (i) All the $m-1$ edges of $T$ must be sampled in $G'$, and,

\,\,\,\,\, (ii) For each $i \in [\numclasses]$ and every class-$i$ vertex $u$ of $T$, at least $\alpha_i$  class-$i$ edges incident to $u$ must appear in $G'$, beyond the class-$i$ edges of $T$ that are incident to $u$.  

It suffices to show that $\Pr[B_T]$ is bounded by \eqref{eq:prob-left-heavy}.

Consider time steps $t=1,\ldots,n$ at which the edges of $G'$ arrive.
For $B_T$ to occur, both the tree edges and the threshold edges must appear during these steps. For each edge in $T$ there are $n$ choices for the time step at which it appears. Next, for each class $i \in [\numclasses]$ and each class-$i$ vertex $u \in T$ , there are $\binom{n}{\alpha_i}$ choices for the $\alpha_i$ time steps when class-$i$ threshold edges incident to $u$ appear. As there are $x_i$ class-$i$ left vertices, the total number of choices  is
\begin{align}\label{eq:sampling-time}
n^{m-1} \; \prod_{i = 1}^\numclasses \binom{n}{\alpha_i}^{x_i} \leq n^{m-1} \; \prod_{i = 1}^\numclasses (en/\alpha_i)^{\alpha_i x_i}. 
\end{align}
Fix such an assignment of edges to time steps. We now bound the probability this is realized in $G'$. The probability that a tree edge is sampled at its assigned time step is $2/(n\avgdegree)$. As a class-$i$ vertex $u$ has degree  at most $\Delta_L(i)$ in $G_i$, the probability that a class-$i$ edge incident to $u$ is sampled at the assigned time step is at most  $2\Delta_L(i)/(n \avgdegree)$.  Putting it all together, this probability is at most  
\begin{align}\label{eq: prob-mapping}
       (2/n\avgdegree)^{m-1} \cdot \prod_{i = 1}^\numclasses \left(2 \Delta_L(i)/(n\avgdegree) \right)^{\alpha_i x_i}.
        \end{align}
Together with \eqref{eq:sampling-time} this gives,
\begin{align}\label{eq: prob-calc}
\Pr[B_T]~\leq~ (\avgdegree/2)^{-(m-1)} \cdot \prod_{i = 1}^\numclasses (2e \Delta_L(i)/(\alpha_i \avgdegree))^{\alpha_i x_i }.
\end{align}
We now plug in the value of $\alpha_i$ and further simplify the expression.  Recall that,
\begin{align}\label{eq: threshold-2}
    \alpha_i = c \cdot  \left(\frac{\maxdensity}{\avgdegree} + \skewness(G)\right) \cdot \left(\frac{\log\!\log n}{2^i}+1\right),
\end{align}
where $c>8$ is a sufficiently large constant.
By \eqref{eq:skewbir-deg} we have $\Delta_L(i)  \leq 16 \maxdensity/f_i $. Also  $\alpha_i  \geq c \maxdensity/\avgdegree$. Thus, $\,\,2e\,\Delta_L(i)/(\alpha_i \avgdegree)\;\le\;1/f_i.$ Plugging this in \eqref{eq: prob-calc} gives,  
\begin{align}\label{eq: prob-cal-2}
\Pr[B_T] \;\le\; (\avgdegree/2) ^{-(m-1)} 
\prod_{i=1}^\numclasses f_i^{-\alpha_i x_i}.
\end{align}
Next we simplify $f_i^{-\alpha_i x_i}$ above. Since $f_i = 2^{2^{i}}$ we have 
$f_i^{(\log\!\log n)/ 2^i} = 2^{\log\!\log n} = \log n$, and thus
$f_i^{((\log\!\log n)/ 2^i + 1)} = f_i \log n$. Therefore, by \eqref{eq: threshold-2},  
\begin{align*}
f_i^{\alpha_i} 
\geq  (f_i \log n) ^ {c \cdot \maxdensity/\avgdegree}  \cdot  (f_i\log n)^{c \cdot \skewness(G)} 
&\geq \left( \maxdensity \log n/ \avgdegree \right)^2   \cdot  (\Delta_R(i)/\maxdensity) .
\end{align*}
Here,  we used the key property of our decomposition that  ${\Delta_R(i) = O(\maxdensity \cdot (f_i \log n)^{2 \skewness(G)})}$ for all $i \in [\numclasses]$.
Moreover, as $f_i\geq 1$ and $\maxdensity/\avgdegree \geq 1/2$, we have 
that $(f_i \log n) ^ {c \maxdensity/\avgdegree}  \geq \left(  \maxdensity \log n/\avgdegree\right)^{2} $ whenever $c > 8$.  Thus, 
 \begin{align}\label{eq: interm}
\prod_{i=1}^\numclasses f_i^{-\alpha_i x_i}
       \leq  (\maxdensity \log n/\avgdegree)^{-2\sum_{i = 1}^\numclasses x_i} \,  \prod_{i = 1}^\numclasses   (\Delta_R(i)/\maxdensity)^{-x_i} 
       \leq (\maxdensity \log n/\avgdegree)^{-m} \,   \prod_{i = 1}^\numclasses   (\Delta_R(i)/\maxdensity)^{-x_i}. 
   \end{align}
Here, the second inequality uses that $T$ is left-heavy and hence ${\sum_{i = 1}^\numclasses x_i \ = \left(\# \text{ Left nodes in $T$} \right) \geq m/2}$. Plugging the bound of \eqref{eq: interm} in \eqref{eq: prob-cal-2}, 
 proves  \eqref{eq:prob-left-heavy}. This completes the proof of  \Cref{clm: pattern} and therefore of \Cref{thm: competitive-ratio}.

\subsection{Reductions for General Number of Arrivals}
\label{sec: reductions}
We now extend \Cref{thm: competitive-ratio} to a general number $T$ of arrivals.

Let $G'_T$ denote the  sampled (multi)-graph after $T$ random edge arrivals, and $\opt(G, T)$ denote the expected optimum load. We note a simple generalization of \Cref{cl:simple-lb-3} for arbitrary $T$, the proof of which follows exactly as in Claim \ref{cl:simple-lb-3}.

\begin{claim}\label{clm: simple-lb}
    Let $G$ be an $n$-vertex base graph with average degree $\avgdegree(G) \geq 1$. Then,
    \begin{align}
    \label{lb:max-density-2}
       \text{For any $T$, }\,\,\,\,\,\,\,\,\,\,\,\,\, \quad \opt(G, T) & \geq 2T\maxdensity(G)/(n\avgdegree(G)). \qquad \qquad \\
    \label{lb:degree-2}
       \text{For any } T \leq n, \,\,\, \quad  \opt(G, T) & = \Omega\left( \log n/\log ((n\avgdegree(G)/T) \cdot \log n) \right) .
    \end{align}
\end{claim}

By a standard doubling trick, we assume that the online algorithm knows $T$. We now consider different regimes of $T$, and show how  $O(\log\!\log n)$ competitiveness follows.

\medskip
\textbf{Case 1 ($T > n\log n$)}: As $\maxdensity(G)/\avgdegree(G) \ge 1/2$ trivially, \eqref{lb:max-density-2} gives that 
$\opt(G,T) = \Omega(\log n)$. 
When  $\opt(G,T)$ is this large, we have the following simple bound, similar to \Cref{cl:smallish-left-degree}. 
\begin{claim}\label{clm:left-degree-bound}
Let $G$ be a left-degree-bounded bipartite graph on $n$ vertices. 
For any number of arrivals $T$, with high probability,  all left vertices have degree 
$O(\opt(G,T) + \log n)$ in the sampled graph $G'_T$.
\end{claim}
So assigning each edge of $G'_T$ to its left endpoint yields a $O(1)$-competitive assignment.

\medskip

\textbf{Case 2 ($n \le T \le n\log n$)}: 
We reduce this setting to the case where the number of arrivals equals the number of vertices in the base graph. 
First, we may assume without loss of generality that $\avgdegree(G) > \log n$. 
Indeed, if $\avgdegree(G) \le \log n$, then by \Cref{cl:simple-lb-3}, even for $n$ arrivals we have $\opt(G,n) = \Omega(\log n / \log\!\log n)$; hence $\opt(G,T) \ge \opt(G,n) = \Omega(\log n / \log\!\log n)$. 
By \Cref{clm:left-degree-bound}, assigning each sampled edge in $G'_T$ to its left endpoint is already $O(\log\!\log n)$-competitive. 

Now suppose $\avgdegree(G) \ge \log n$. 
We construct an augmented graph $G^{\text{aug}}$ by adding $T-n$ isolated vertices to $G$, so that $G^{\text{aug}}$ has $T$ vertices. 
Since the isolated vertices contribute no edges, sampling $T$ random edges from $G^{\text{aug}}$ is equivalent to sampling $T$ random edges from $G$; consequently, $\opt(G^{\text{aug}}, T) = \opt(G, T)$. 
Moreover, $\avgdegree(G^{\text{aug}}) \ge \avgdegree(G)/\log n \ge 1$. 

Applying \Cref{thm: competitive-ratio} to $G^{\text{aug}}$ for $T$ i.i.d.\ arrivals gives expected maximum load 
$O(\log\!\log T \cdot \opt(G^{\text{aug}}, T)) = O(\log\!\log n \cdot \opt(G, T))$, 
thus giving an $O(\log\!\log n)$-competitive algorithm.

\medskip
\textbf{Case 3 ($\log n < T < n$)}:  Let $\gamma := T/n < 1$, so that the number of arrivals is $T = \gamma n$.  

If $\avgdegree(G)/\gamma = O(\log n)$, then \Cref{lb:degree-2} implies
$\opt(G,T)=\Omega(\log n/\log\!\log n)$, and so by \Cref{clm:left-degree-bound},
assigning each sampled edge to its left endpoint is already
$O(\log\!\log n)$-competitive. Hence, we may assume
$\avgdegree(G)/\gamma = \Omega(\log n)$. 

We now reduce this setting to the case where the number of arrivals equals the number of vertices.

Construct an augmented graph $G^{\text{aug}} = G \cup H$, where $H$ is a union of $\lceil n / \avgdegree(G) \rceil$ disjoint cliques, each of size $k = \lceil \avgdegree(G) / \gamma \rceil$, with vertex sets disjoint from $G$.  
The augmented graph has $N = \Theta(n / \gamma)$ vertices in total, and only a $\Theta(\gamma^2)$ fraction of its edges belong to $G$.

If $N$ edges are sampled uniformly at random from $G^{\text{aug}}$, whp, the number of edges sampled from  $G$ is $\Theta(\gamma^2 N) = \Theta(\gamma n) = \Theta(T)$. Further, the number of edges sampled from each clique, is $\Theta(k)$ whp (as $k = \lceil{\avgdegree(G)/\gamma\rceil} =\Omega(\log n))$.  So the edges within each clique can be oriented so that the load is~$O(1)$, giving that $\opt(G^{\text{aug}}, N) = O(\opt(G, T))$.

By \Cref{thm: competitive-ratio}, there is an algorithm for $N$ arrivals in  $G^{\text{aug}}$, which ensures that the expected maximum load on any vertex in $G$ is $O(\log\!\log N \cdot \opt(G^{\text{aug}}, N)) = O(\log\!\log n \cdot \opt(G, T))$.   Using the same preprocessing, we can run this algorithm directly on $G$ for the original $T$ random edge arrivals; since the auxiliary cliques in $H$ are disjoint from $G$, this preserves the same load guarantees. Hence we obtain an $O(\log\!\log n)$-competitive algorithm for $T$ arrivals in $G$, completing the reduction.

\medskip
\textbf{Case 4 ($T \le \log n$)}: 
Here, the Greedy algorithm is $O(\log\!\log n)$-competitive by \Cref{fact:greedy}.

\section{Lower Bound for Greedy}\label{sec: greedy-lb}

We now prove that the Greedy algorithm has an $\Omega(\log n/ (\log\!\log n)^2)$-competitive ratio, even for some mildly irregular base graphs.  

\textbf{The Graph.}
The graph $G$ is layered, with vertices partitioned as $V_1, \ldots, V_b$ (see \Cref{fig:greedy-lower-bound}).
Let $t = (\log n)^3$ and set $|V_i| = t^i \sqrt{n}$ for each $i \in [b]$.
We choose $b = O(\log n / \log\!\log n)$ so that the total number of vertices is $\Theta(n)$.
For each $i \in [b-1]$, we construct the edges between $V_i$ and $V_{i+1}$ as follows:
partition $V_i$ into $t^i$ groups of size $\sqrt{n}$, and $V_{i+1}$ into $t^i$ groups of size $\sqrt{n}t$.
Between each corresponding pair of groups, we place a complete bipartite graph $K_{\sqrt{n},\,\sqrt{n}t}$.
Let $G_i$ denote the bipartite subgraph of $G$ induced by $V_i \cup V_{i+1}$;  
it is biregular, with every vertex in $V_i$ having degree $\sqrt{n}t$ and every vertex in $V_{i+1}$ having degree $\sqrt{n}$.  
The total number of edges is $\sum_{i=1}^{b-1} t^i \cdot \sqrt{n}t \cdot \sqrt{n} = \Theta(n t^b) = \Theta(n^{1.5})$, so the average degree is $\Theta(\sqrt{n})$.

\begin{figure*}[t]
\centering
\begin{tikzpicture}[
    x=1cm,y=1cm,
    >=Latex,
    box/.style={rounded corners=2pt, draw=black, thick, fill=gray!10},
    mini/.style={rounded corners=2pt, draw=black, thick},
    grp/.style={rounded corners=1pt, draw=black!70, fill=gray!15},
    note/.style={font=\footnotesize, align=center},
    lab/.style={font=\small, align=center},
    badge/.style={circle, draw=black, fill=white, inner sep=1.4pt, font=\scriptsize},
    bundle/.style={gray!70, line width=0.35pt},
    flow/.style={-Latex, thick}
]


\draw[box] (-0.38,-0.65) rectangle (0.38,0.65);
\node at (0,0) {$V_1$};
\node[note] at (0,-1.0) {$t\sqrt n$};

\draw[box] (1.52,-0.85) rectangle (2.28,0.85);
\node at (1.9,0) {$V_2$};
\node[note] at (1.9,-1.2) {$t^2\sqrt n$};

\draw[box] (3.42,-1.05) rectangle (4.18,1.05);
\node at (3.8,0) {$V_3$};
\node[note] at (3.8,-1.4) {$t^3\sqrt n$};

\node[lab] at (5.1,0) {$\cdots$};

\draw[box] (6.62,-1.25) rectangle (7.38,1.25);
\node at (7.0,0) {$V_{b-2}$};
\node[note] at (7.0,-1.6) {$t^{b-2}\sqrt n$};

\draw[box] (8.52,-1.45) rectangle (9.28,1.45);
\node at (8.9,0) {$V_{b-1}$};
\node[note] at (8.9,-1.8) {$t^{b-1}\sqrt n$};

\draw[box] (10.42,-1.65) rectangle (11.18,1.65);
\node at (10.8,0) {$V_b$};
\node[note] at (10.8,-2.0) {$t^{b}\sqrt n$};

\foreach \y in {-0.45,-0.15,0.15,0.45}{
  \draw[bundle] (0.38,\y) -- (1.52,{0.85/0.65*\y});
  \draw[bundle] (2.28,{0.85/0.65*\y}) -- (3.42,{1.05/0.85*\y});
  \draw[bundle] (7.38,{1.25/1.65*\y}) -- (8.52,{1.45/1.65*\y});
  \draw[bundle] (9.28,{1.45/1.65*\y}) -- (10.42,\y);
}
\draw[bundle,densely dashed] (4.25,0.45) -- (6.55,0.45);
\draw[bundle,densely dashed] (4.25,0.00) -- (6.55,0.00);
\draw[bundle,densely dashed] (4.25,-0.45) -- (6.55,-0.45);

\node[badge] at (10.8,2.1) {$0$};
\draw[flow] (10.45,2.1) -- (9.25,2.1);
\node[note] at (9.85,2.42) {$B_1$};

\node[badge] at (8.9,2.1) {$1$};
\draw[flow] (8.55,2.1) -- (7.35,2.1);
\node[note] at (7.95,2.42) {$B_2$};

\node[badge] at (7.0,2.1) {$2$};
\draw[flow,densely dashed] (6.65,2.1) -- (5.25,2.1);

\node[note] at (5.0,-2.5) {Each vertex gets $\Omega(t/\log n)=\Omega(\log^2 n)$ relevant arrivals per batch.};

\end{tikzpicture}
\captionsetup{justification=raggedright,singlelinecheck=false}
\caption{Lower-bound construction for Greedy. The graph has layers $V_1,\ldots,V_b$ with
$|V_i|=t^i\sqrt n$, where $t=(\log n)^3$, and between consecutive layers the edges form
$t^i$ disjoint copies of $K_{\sqrt n,\,\sqrt n t}$. The figure shows the batch-wise
propagation of load by Greedy: batch $B_1$ pushes all of $V_{b-1}$ to load at least $1$, batch $B_2$ pushes all of
$V_{b-2}$ to load at least $2$, and so on.}
\label{fig:greedy-lower-bound}
\end{figure*}

We now consider $n$ i.i.d uniformly random arrivals from $G$. We first show a lower bound on the load when the  Greedy algorithm is used to orient these edges. We then  show that $\opt(G)$  is $O(\log\!\log n)$.

\begin{lemma}\label{lem: greedy-high-abs-load}
    The Greedy Algorithm, with random tiebreaking, has max-load $\Omega(\log n/\log\!\log n)$ with high probability.
\end{lemma}
\begin{proof}
Let $e_1, \ldots, e_n$ be the edges in the sampled graph $G'$ in the order they arrive. Group these edges into batches $B_1,\ldots,B_{\log n}$ of  $k=n/\log n$ edges, where $B_i = \{e_j \, : (i-1) k < j \leq ik \} $.
We will show at the end of batch $i$, whp, all the vertices in $V_{b - i}$ (layer $b-i$ of $G$) have load at least $i$. This implies the lemma as eventually vertices at level $1$ will have load $\Omega(\log n/ \log\!\log n)$.

For $i = 0, \ldots, b-1$, let $E_i$ denote the event that all vertices in $V_{b-i}$ have load at least $i$ after Greedy processes the first $i$ batches. 
We will show that $\Pr[E_i] \geq (1 - i/n^2)$, by induction.

First we make an observation. Fix some vertex $u$ in $V_i$ for $i \in [b - 1]$. As  $\avgdegree(G)=\Theta(\sqrt{n})$, and $u$ has $t \sqrt{n}$ edges to $V_{i+1}$, during each batch, in expectation, $\Omega(t/\log n ) = \Omega(\log^2 n)$ edges will arrive, and this occurs with probability at least $1 - 1/\poly(n)$ by Chernoff bounds.

The base case is trivial as all vertices in $V_b$ have load $0$ before the edges arrive; thus, $\Pr[E_0] = 1$. 

Suppose, inductively, that $\Pr[E_{i-1}] \geq (1 - (i-1)/n^2)$. To show that $\Pr[E_i] \geq (1 - i/n^2)$  it suffices to show that $\Pr[\overline{E}_i|E_{i-1}] \leq 1/n^2$, as $\Pr[\overline{E}_i] \leq \Pr[\overline{E}_i| E_{i-1}] + \Pr[\overline{E}_{i-1}]$.

Consider a vertex $u$ at level $V_{b - i}$. Arguing as before, whp, batch $B_i$ contains $\Omega(\log ^2 n)$ edges from $u$ to $V_{b-i+1}$. Call this set of edges $S$. Condition on the event $E_{i-1}$.  After the first $i-1$ edges of $S$ arrive,  Greedy ensures that the load at $u$ is at least $i-1$.
Once the load of $u$ reaches $(i-1)$, each subsequent edge of $S$ has probability at least $1/2$ of being assigned to $u$. Thus, by the end of this batch,  the load at $u$ is at least $i$ with probability at least  $1 - 2^{-\Omega(\log^{2} n)}$.  A union bound over the vertices in $V_{b - i}$ gives that $\Pr[\overline{E}_i| E_{i-1}] \leq 1/n^2$. 
\end{proof}

A simple argument shows that the optimum expected load for $G$ is $O(\log\!\log n)$.
\begin{lemma}\label{lem: gs-low-density}
$\opt(G) = O(\log\!\log n)$.
\end{lemma}
\begin{proof}
Recall that $G$ consists of $O(\sqrt{n})$ disjoint copies of $K_{\sqrt{n},\,\sqrt{n}t}$.
We will show that the max-density of each such copy in the sampled graph $G'$ is  w.h.p.\ $O(\log\!\log n)$. As 
each vertex of $G$ lies in at most two copies of $K_{\sqrt{n},\,\sqrt{n}t}$, this would give that the total load at any vertex is $O(\log\!\log n)$.

Using that each edge of $G'$ is sampled with probability $\Theta(1/\sqrt{n})$, the claimed $O(\log\!\log n)$ bound on the max-density of  $K_{\sqrt{n},\,\sqrt{n}t}$ in $G'$ follows  by considering all its subgraphs with $a \leq  \sqrt{n}$ left vertices and $b \leq \sqrt{n} t $ right vertices, using standard probability tail bounds on the number of edges, and applying a union bound.
 \qedhere  
 
\end{proof}

\section{Conclusion and Future Directions}\label{sec:conclusion}

We gave an $O(\log\!\log n)$-competitive algorithm for online graph balancing under i.i.d.\ arrivals from an arbitrary base graph known in advance, matching the $\Omega(\log\!\log n)$ lower bound for complete graphs up to constant factors. Thus, the i.i.d.\ setting admits a qualitatively stronger guarantee than the $\Theta(\log n)$ bound for adversarial arrivals. Conceptually, our work identifies \emph{\logskew{}} as the key obstruction on irregular graphs and shows how bounded \logskew{} yields a decomposition into only $O(\log\!\log n)$ \skewbir{} graphs on which a greedy-style algorithm succeeds. We conclude with two concrete directions for future work.

\begin{itemize}
    \item \textbf{Direction 1: From graph balancing to general load balancing.}
    Our work focuses on graph balancing, where each job corresponds to an edge that can be assigned to one of its two incident machines/vertices. A major generalization is the classical \emph{unrelated-machines} setting with $n$ machines, where each job $j$ may be assigned to any machine $i$ and has an arbitrary processing time $p_{ji} \ge 0$. In the adversarial setting, this problem admits a $\Theta(\log n)$-competitive algorithm~\cite{aspnes1997line}. A natural question is whether stochastic arrivals from a known distribution over processing-time vectors allow an $O(\poly(\log\!\log n))$-competitive ratio. In particular, this would require extending our techniques from graphs to hypergraphs.

    \item \textbf{Direction 2: Beyond known base graphs.}
    Our $O(\log\!\log n)$ result relies crucially on knowing the base graph in advance. It would be interesting to understand how far these guarantees extend when this structure is only partially known. This is especially intriguing in light of the recent lower bound of~\cite{im2024online} for the related \emph{random-order model}, where the multiset of edges may be adversarially chosen and only the arrival order is random: every algorithm in that model has competitive ratio $\Omega(\sqrt{\log n})$. This contrast suggests that the full power of the i.i.d.\ model may depend on access to some structural information about the base graph. Can one still obtain $O(\log\!\log n)$ guarantees in the i.i.d.\ model when the base graph is {unknown}, or in sample-based models such as AOS where the algorithm receives only a random partial sample of an adversarially generated sequence~\cite{kaplan2022online,AFGS-22}? Understanding what structural information is really necessary to recover the benefits of i.i.d.\ arrivals remains an appealing direction for future work.
\end{itemize}

\clearpage

\begin{small}

\begin{thebibliography}{CMadHS95}

\bibitem[AAF{\etalchar{+}}97]{aspnes1997line}
James Aspnes, Yossi Azar, Amos Fiat, Serge Plotkin, and Orli Waarts.
\newblock On-line routing of virtual circuits with applications to load balancing and machine scheduling.
\newblock {\em Journal of the ACM, JACM}, 44(3):486--504, 1997.

\bibitem[ABKU94]{azar1994balanced}
Yossi Azar, Andrei Broder, Anna Karlin, and Eli Upfal.
\newblock Balanced allocations.
\newblock In {\em Symposium on Theory of Computing, {STOC}}, pages 593--602, 1994.

\bibitem[AFGS22]{AFGS-22}
C.~J. Argue, Alan~M. Frieze, Anupam Gupta, and Christopher Seiler.
\newblock Learning from a sample in online algorithms.
\newblock In {\em Annual Conference on Neural Information Processing Systems, NeurIPS}, 2022.

\bibitem[ANR92]{ABK-FOCS92}
Yossi Azar, Joseph Naor, and Raphael Rom.
\newblock The competitiveness of on-line assignments.
\newblock In {\em Symposium on Discrete Algorithms, {SODA}}, pages 203--210, 1992.

\bibitem[Aza05]{azar2005line}
Yossi Azar.
\newblock On-line load balancing.
\newblock {\em Online algorithms: the state of the art}, pages 178--195, 2005.

\bibitem[BCSV00]{berenbrink2000balanced}
Petra Berenbrink, Artur Czumaj, Angelika Steger, and Berthold V{\"o}cking.
\newblock Balanced allocations: The heavily loaded case.
\newblock In {\em Symposium on Theory of Computing, {STOC}}, pages 745--754, 2000.

\bibitem[BF22]{bansal2022power}
Nikhil Bansal and Ohad Feldheim.
\newblock The power of two choices in graphical allocation.
\newblock In {\em Symposium on Theory of Computing, {STOC}}, pages 52--63, 2022.

\bibitem[BFL{\etalchar{+}}95]{BroderFLPR-IPL95}
Andrei Broder, Alan Frieze, Carsten Lund, Steven Phillips, and Nick Reingold.
\newblock Balanced allocations for tree-like inputs.
\newblock {\em Inf. Process. Lett.}, 55(6):329--332, 1995.

\bibitem[BK22]{bansal2022balanced}
Nikhil Bansal and William Kuszmaul.
\newblock Balanced allocations: The heavily loaded case with deletions.
\newblock In {\em Foundations of Computer Science, {FOCS}}, pages 801--812, 2022.

\bibitem[Car08]{Caragiannis-SODA08}
Ioannis Caragiannis.
\newblock Better bounds for online load balancing on unrelated machines.
\newblock In {\em Symposium on Discrete Algorithms, SODA}, pages 972--981, 2008.

\bibitem[CFM{\etalchar{+}}98]{cole1998balls}
Richard Cole, Alan Frieze, Bruce Maggs, Michael Mitzenmacher, Andr{\'e}a Richa, Ramesh Sitaraman, and Eli Upfal.
\newblock On balls and bins with deletions.
\newblock In {\em Randomization and Approximation Techniques in Computer Science, {RANDOM}}, pages 145--158, 1998.

\bibitem[CMadHS95]{czumaj1995shared}
Artur Czumaj, Friedhelm Meyer auf~der Heide, and Volker Stemann.
\newblock Shared memory simulations with triple-logarithmic delay.
\newblock In {\em European Symposium on Algorithms}, pages 46--59. Springer, 1995.

\bibitem[DR96]{dubhashi1996balls}
Devdatt~P Dubhashi and Desh Ranjan.
\newblock Balls and bins: A study in negative dependence.
\newblock {\em BRICS Report Series}, 3(25), 1996.

\bibitem[GM16]{GuptaM-MOR16}
Anupam Gupta and Marco Molinaro.
\newblock How the experts algorithm can help solve lps online.
\newblock {\em Math. Oper. Res.}, 41(4):1404--1431, 2016.

\bibitem[GMP20]{greenhill2020balanced}
Catherine Greenhill, Bernard Mans, and Ali Pourmiri.
\newblock Balanced allocation on hypergraphs.
\newblock {\em arXiv preprint arXiv:2006.07588}, 2020.

\bibitem[God08]{godfrey2008balls}
Brighten Godfrey.
\newblock Balls and bins with structure: balanced allocations on hypergraphs.
\newblock In {\em Symposium on Discrete Algorithms, {SODA}}, pages 511--517, 2008.

\bibitem[Hak65]{hakimi1965degrees}
S~Louis Hakimi.
\newblock On the degrees of the vertices of a directed graph.
\newblock {\em Journal of the Franklin Institute}, 279(4):290--308, 1965.

\bibitem[IKL{\etalchar{+}}24]{im2024online}
Sungjin Im, Ravi Kumar, Shi Li, Aditya Petety, and Manish Purohit.
\newblock Online load and graph balancing for random order inputs.
\newblock In {\em Symposium on Parallelism in Algorithms and Architectures, {SPAA}}, pages 491--497, 2024.

\bibitem[KMPS25]{KMPS-FOCS25}
Thomas Kesselheim, Marco Molinaro, Kalen Patton, and Sahil Singla.
\newblock Integral online algorithms for set cover and load balancing with convex objectives.
\newblock In {\em Foundations of Computer Science, {FOCS}}, 2025.

\bibitem[KMS23]{KMS-SODA23}
Thomas Kesselheim, Marco Molinaro, and Sahil Singla.
\newblock Online and bandit algorithms beyond $\ell_p$ norms.
\newblock In {\em Symposium on Discrete Algorithms {SODA}}, pages 1566--1593. {SIAM}, 2023.

\bibitem[KNR22]{kaplan2022online}
Haim Kaplan, David Naori, and Danny Raz.
\newblock Online weighted matching with a sample.
\newblock In {\em Symposium on Discrete Algorithms, SODA}, pages 1247--1272, 2022.

\bibitem[Knu97]{knuth1997art}
Donald Knuth.
\newblock {\em The art of computer programming}, volume~3.
\newblock Pearson, 1997.

\bibitem[KP06]{kenthapadi2006balanced}
Krishnaram Kenthapadi and Rina Panigrahy.
\newblock Balanced allocation on graphs.
\newblock In {\em Symposium on Discrete Algorithms, {SODA}}, pages 434--443, 2006.

\bibitem[LPY19]{lenzen2019parallel}
Christoph Lenzen, Merav Parter, and Eylon Yogev.
\newblock Parallel balanced allocations: The heavily loaded case.
\newblock In {\em Symposium on Parallelism in Algorithms and Architectures, {SPAA}}, pages 313--322, 2019.

\bibitem[LS21]{los2021balanced}
Dimitrios Los and Thomas Sauerwald.
\newblock Balanced allocations with incomplete information: The power of two queries.
\newblock {\em arXiv preprint arXiv:2107.03916}, 2021.

\bibitem[LS23]{los2023balanced}
Dimitrios Los and Thomas Sauerwald.
\newblock Balanced allocations with the choice of noise.
\newblock {\em Journal of the ACM}, 70(6):1--84, 2023.

\bibitem[LX21]{LiXian-ICML21}
Shi Li and Jiayi Xian.
\newblock Online unrelated machine load balancing with predictions revisited.
\newblock In {\em International Conference on Machine Learning, {ICML}}, pages 6523--6532, 2021.

\bibitem[Mit96]{mitzenmacher1996power}
Michael~David Mitzenmacher.
\newblock {\em The Power of Two Choices in Randomized Load Balancing}.
\newblock PhD thesis, University of California at Berkeley, 1996.

\bibitem[Mol17]{molinaro2017online}
Marco Molinaro.
\newblock Online and random-order load balancing simultaneously.
\newblock In {\em Symposium on Discrete Algorithms, {SODA}}, pages 1638--1650. SIAM, 2017.

\bibitem[MRS01]{sitaraman2001power}
M.~Mitzenmacher, A.~Richa, and R.~Sitaraman.
\newblock The power of two random choices: A survey of techniques and results.
\newblock {\em Handbook of Randomized Computing: volume 1, edited by P. Pardalos, S. Rajasekaran, and J. Rolim}, 2001.

\bibitem[MVLL20]{lattanzi2020learning}
Benjamin Moseley, Sergei Vassilvitskii, Silvio Lattanzi, and Thomas Lavastida.
\newblock Online scheduling via learned weights.
\newblock In {\em Symposium on Discrete Algorithms, {SODA}}, 2020.

\bibitem[PTW15]{peres2015graphical}
Yuval Peres, Kunal Talwar, and Udi Wieder.
\newblock Graphical balanced allocations and the (1+ $\beta$)-choice process.
\newblock {\em Random Structures \& Algorithms}, 47(4):760--775, 2015.

\bibitem[Rou21]{Roughgarden-BeyondBook}
Tim Roughgarden.
\newblock {\em Beyond the worst-case analysis of algorithms}.
\newblock Cambridge University Press, 2021.

\bibitem[Ste96]{stemann1996parallel}
Volker Stemann.
\newblock Parallel balanced allocations.
\newblock In {\em Symposium on Parallel algorithms and architectures, {SPAA}}, pages 261--269, 1996.

\bibitem[TW14]{talwar2014balanced}
Kunal Talwar and Udi Wieder.
\newblock Balanced allocations: A simple proof for the heavily loaded case.
\newblock In {\em International Colloquium on Automata, Languages, and Programming, ICALP}, pages 979--990, 2014.

\bibitem[V{\"o}c03]{vocking2003asymmetry}
Berthold V{\"o}cking.
\newblock How asymmetry helps load balancing.
\newblock {\em Journal of the ACM, JACM}, 50(4):568--589, 2003.

\bibitem[Wie17]{wieder2017hashing}
Udi Wieder.
\newblock Hashing, load balancing and multiple choice.
\newblock {\em Foundations and Trends{\textregistered} in Theoretical Computer Science}, 12(3--4):275--379, 2017.

\end{thebibliography}
\newcommand{\etalchar}[1]{$^{#1}$}

\end{small}

\appendix

\section{Proofs of Simple Lower Bounds on \texorpdfstring{$\opt(G)$}{OPT(G)}}
In this section, we will prove   \Cref{cl:simple-lb-3} where we show simple lower bounds on $\opt(G)$ due to dense components and due to low average degree. Towards this end we recall some standard bounds on the tails of the binomial distribution.

\label{sec:bin-conc}

\begin{fact}
 \label{fact:binom}
     Let $X \sim \text{Bin}(n,p)$ with mean $\mu=np$. Then,
      \begin{enumerate}
          \item 
     $\Pr[X \geq k\mu] \leq (e/k)^{k\mu}$ for any $k \geq 3$.   
    \item If $p=O(1/n)$,
for any $\delta = \exp(-o(n^{1/2}))$, we have 
$X = \Omega(\log(1/\delta)/(\log(1/\mu) + \log  \log(1/\delta)))$ with probability at least $\delta$.\footnote{This follows as
$\Pr[X \geq k] = \Theta (\mu^k/k!)$ for $k = o(\sqrt{n})$.
}
     \end{enumerate}
 \end{fact}

\begin{proof}[Proof of \Cref{cl:simple-lb-3}]
As each edge of $G$ appears $2/\avgdegree(G)$ times  in $G'$ on average, for any subset $S\subset V$, the expected density
$\E[ \rho(G'[S])]$ in $G'$ is simply $(2/\avgdegree(G)) \, \rho(G[S])$. Thus by \eqref{eq:opt}, 
\[  \opt(G)  = \E_{G'} \Big[ \max_{S \subseteq V} \rho(G'[S])\Big] \geq \max_{S \subseteq V} \E_{G'} [\rho(G'[S])] =   (2/\avgdegree) \, \rho(G[S]).\]
The second  bound follows by noting that some edge is likely to appear 
$k  = \Omega(\log n/ (\log(\avgdegree \cdot \log n)))$ times in $G'$, and hence one of its endpoints must have load $\geq k/2$. 
Indeed, fix some edge $e$, and let $X_e$ denote the number of its occurences in $G'$. As $X_e\sim \text{Bin}(n,1/|E|)$, Fact \ref{fact:binom} with $\delta=10/|E|$, gives that $\Pr[X_e\geq k] \geq \delta$.
So, in expectation, at least $10$ edges of $G$ appear $\geq k$ times in $G'$, and by a standard second moment argument this holds for some edge with  probability $\Omega(1)$.
\end{proof}

\end{document}